\documentclass{theoretics}

\title{Approximate counting for spin systems in sub-quadratic time}

\ThCSauthor[affil1, affil2]{ Konrad Anand}{ konrad.anand@me.com }[0000-0001-5778-9397]
\ThCSauthor[affil3]{ Weiming Feng}{ fwm1994@gmail.com }[0000-0003-4636-1023]
\ThCSauthor[affil2]{ Graham Freifeld}{ g.freifeld@sms.ed.ac.uk }[0009-0007-2471-5413]
\ThCSauthor[affil2]{ Heng Guo}{ hguo@inf.ed.ac.uk }[0000-0001-8199-5596]
\ThCSauthor[affil4]{ Jiaheng Wang}{ pw384@hotmail.com }[0000-0002-5191-545X]

\ThCSaffil[affil1]{ School of Mathematical Sciences, Queen Mary University of London, United Kingdom }
\ThCSaffil[affil2]{ School of Informatics, University of Edinburgh, United Kingdom }
\ThCSaffil[affil3]{ Institute for Theoretical Studies, ETH Z\"{u}rich, Z\"{u}rich, Switzerland }
\ThCSaffil[affil4]{ Faculty of Informatics and Data Science, University of Regensburg, Germany }

\ThCSthanks{This project has received funding from the European Research Council (ERC) under the European Union's Horizon 2020 research and innovation programme (grant agreement No.~947778). Weiming Feng acknowledges the support from Dr.~Max R\"ossler, the Walter Haefner Foundation and the ETH Z\"urich Foundation. A preliminary version of this article appeared at ICALP 2024 \cite{AnandFFGW24}.}
\ThCSshortnames{K. Anand, W. Feng, G. Freifeld, H. Guo, J. Wang}  
\ThCSshorttitle{Approximate counting for spin systems in sub-quadratic time}
\ThCSyear{2025}
\ThCSarticlenum{3}
\ThCSreceived{May 31, 2024}
\ThCSrevised{Nov 4, 2024}
\ThCSaccepted{Nov 17, 2024}
\ThCSpublished{Jan 13, 2025}
\ThCSdoicreatedtrue
\ThCSkeywords{Randomised algorithm, Approximate counting, Spin system, Sub-quadratic algorithm}

\usepackage{mathtools} 

\usepackage[thinlines]{easytable}

\usepackage{pgfplots}
\pgfplotsset{compat=1.16}
\usepackage{tikz}
\usetikzlibrary{arrows,arrows.meta,backgrounds,calc,fit,decorations.pathreplacing,decorations.markings,shapes.geometric}

\tikzstyle{internal} = [draw, fill, shape=circle]
\tikzstyle{external} = [shape=circle]
\tikzstyle{square}   = [draw, fill, rectangle]
\tikzstyle{triangle} = [draw, fill, regular polygon, regular polygon sides=3, inner sep=3pt]
\tikzstyle{pentagon} = [draw, fill, regular polygon, regular polygon sides=5, inner sep=2pt, minimum size=14pt]
\tikzset{every fit/.append style=text badly centered}

\usetikzlibrary{positioning,chains,fit,shapes,calc}
\usetikzlibrary{trees}
\usetikzlibrary{decorations.pathreplacing}
\usetikzlibrary{decorations.pathmorphing}
\usetikzlibrary{decorations.markings}
\tikzset{>=latex} 

\usepackage{ifthen}

\usepackage[capitalise,noabbrev,nameinlink]{cleveref}

\usepackage[textsize=tiny]{todonotes}

\usepackage[normalem]{ulem}

\newcommand{\tp}[1]{{\left( #1 \right)}}

\newcommand{\Ex}{\mathop{\mathbb{{}E}}\nolimits}
\renewcommand{\Pr}{\mathop{\mathrm{Pr}}\nolimits}

\def\+#1{\mathcal{#1}}
\def\=#1{\mathbb{#1}}

\newcommand{\poly}[1]{\ensuremath{\mathop{\mathrm{poly}}\tp{#1}}}
\newcommand{\polylog}[1]{\ensuremath{\mathop{\mathrm{polylog}}\tp{#1}}}

\newcommand{\norm}[2]{\ensuremath{\Vert #2 \Vert_{#1}}}
\newcommand{\abs}[1]{\ensuremath{\left\vert#1\right\vert}}

\newcommand{\ceil}[1]{\lceil#1\rceil}

\newcommand{\eps}{\varepsilon}

\newcommand{\Var}[2]{\ensuremath{\textnormal{Var}_{#1}\left(#2\right)}}

\newcommand{\dist}{\operatorname{dist}}

\newcommand{\defeq}{:=}

\newcommand{\numP}{\#{\textnormal{\textbf{P}}}}
\newcommand{\NP}{\textnormal{\textbf{NP}}}

\makeatletter
\def\prob#1#2#3{\goodbreak\begin{list}{}{\labelwidth\z@ \itemindent-\leftmargin
      \itemsep\z@  \topsep6\p@\@plus6\p@
      \let\makelabel\descriptionlabel}
  \item[\textbf{Name}]#1
  \item[\textbf{Instance}]#2
  \item[\textbf{Output}]#3
  \end{list}}
\makeatother

\makeatletter
\providecommand\@dotsep{5}
\def\listtodoname{Todo list}
\def\listoftodos{\@starttoc{tdo}\listtodoname}
\makeatother

\newcommand{\hcsampler}{\mathsf{HardcoreSampler}}
\newcommand{\ssmsampler}{\mathsf{LazySampler}}
\newcommand{\bdsplit}{\mathsf{BoundarySplit}}
\newcommand{\dTV}{d_{\mathrm{TV}}}
\newcommand{\TSAW}{T_{\mathrm{SAW}}}
\newcommand{\tw}{\operatorname{tw}}
\newcommand{\argmax}{\operatorname{argmax}}

\begin{document}
\maketitle

\begin{abstract}
  We present two randomised approximate counting algorithms with running time $\widetilde{O}\left(\left(\frac{n}{\eps}\right)^{2-c}\right)$, for some constant $c>0$ and accuracy $\eps$:
  \begin{enumerate}
    \item for the hard-core model with fugacity $\lambda$ on graphs with maximum degree $\Delta$ when $\lambda=O(\Delta^{-1.5-c_1})$ where $c_1=c/(2-2c)$;
    \item for spin systems with strong spatial mixing (SSM) on planar graphs with quadratic growth, such as $\mathbb{Z}^2$.
  \end{enumerate}

  For the hard-core model, Weitz's algorithm (STOC, 2006) achieves sub-quadratic running time when correlation decays faster than the neighbourhood growth, namely when $\lambda = o(\Delta^{-2})$.
  Our first algorithm does not require this property and extends the range where sub-quadratic algorithms exist.

  Our second algorithm appears to be the first to achieve sub-quadratic running time up to the SSM threshold, albeit on a restricted family of graphs.
  It also extends to (not necessarily planar) graphs with polynomial growth, such as $\mathbb{Z}^d$,
  but with a running time of the form $\widetilde{O}\left(\left(\frac{n}{\eps}\right)^{2}/2^{c(\log \frac{n}{\eps})^{1/d}}\right)$ where $d$ is the exponent of the polynomial growth and $c>0$ is some constant.
\end{abstract}

\section{Introduction}

The study of counting complexity was initiated by Valiant \cite{Val79} with the introduction of the complexity class $\numP$.
An intriguing phenomenon emerging in counting complexity is that many $\numP$-complete problems admit fully polynomial-time randomised approximation schemes (FPRAS),
which output an $\eps$-approximation in time polynomial in $n$ and $1/\eps$ with $n$ being the input size.
This is most commonly found for the so-called partition function of spin systems, as demonstrated by the pioneering work of Jerrum and Sinclair \cite{JS89,JS93}.
Spin systems are physics models for nearest neighbour interactions, and the partition functions are the normalising factors for their Gibbs distributions.
This quantity can express the count of combinatorial objects such as the number of matchings, independent sets, or colourings in a graph,
and is much more expressible by allowing real parameters of the system.

In this paper we are most interested in the fine-grained aspects of the complexity of estimating partition functions.
While for most spin systems, exact counting is $\numP$-hard \cite{CC17},
there are parameter ranges where the model exhibits spatial mixing\footnote{There are multiple notions of spatial mixing with varying degrees of strength. In this paper we use strong spatial mixing, defined in \Cref{sec:prelim}, \Cref{def:SSM}. } that yields efficient approximation algorithms. 
Roughly speaking, spatial mixing gives us a way to bound the correlation or influence a partial configuration has on each vertex with respect to its distance to the vertex. 
Without spatial mixing, the partition function is usually $\NP$-hard to approximate \cite{SS14,GSV16,GSV15}.

Efficient approximate counting was first enabled by the work of Jerrum, Valiant, and Vazirani \cite{JVV86} 
who gave self-reductions from approximate counting to sampling for a large class of problems.
The sampling task is then most commonly solved via Markov chains.
The efficiency of a Markov chain is measured by its mixing time (i.e.~how long it takes to get close to the target distribution).
For spin systems with spatial mixing, in many situations, the standard chain, namely the Glauber dynamics, mixes in $O(n\log n)$ time \cite{CLV21,AJKPV22,CFYZ22,CE22}.

Another later technique, simulated annealing, provides a more efficient counting to sampling reduction \cite{SVV09,Hub15,Kol18}.
Together with the $O(n\log n)$ mixing time mentioned above, this leads to\footnote{The notation $\widetilde{O}(\cdot)$ hides logarithmic factors in $n$ and sometimes other constants, such as ones depending on the maximum degree $\Delta$ of the input graph.} $\widetilde{O}((n/\eps)^2)$ approximate counting algorithms. 
These Markov chain Monte Carlo (MCMC) algorithms are the fastest for estimating partition functions in general,
but $\Omega(n^2)$ appears to be a natural barrier to this approach.
This is because generating a sample would take at least linear time (and there are $\Omega(n\log n)$ lower bounds for the mixing time of Markov chains \cite{HS07} for many spin systems),
and, restricted to the standard way of using the samples, the number of samples required for simulated annealing is at least $\Omega(n/\eps^2)$ \cite[Theorem~10]{Kol18}.

On the other hand, when we relax the parameters, $\Omega(n^2)$ is no longer a barrier to algorithms.
Let us take the hard-core gas model as an example. 
Here the Gibbs distribution $\mu$ is over the set $\+I$ of independent sets of a graph $G$.
For an independent set $I$, $\mu(I)\defeq\lambda^{\abs{I}}/Z(G)$, where $\lambda$ is a parameter of the system (so-called fugacity),
and $Z(G)\defeq\sum_{I\in\+I} \lambda^{\abs{I}}$ is the partition function.
For graphs with degree bound $\Delta$, spatial mixing holds when $\lambda<\lambda_c(\Delta) \defeq \frac{(\Delta-1)^{\Delta-1}}{(\Delta-2)^{\Delta}} \approx\frac{e}{\Delta}$.
The aforementioned MCMC results \cite{CLV21,CFYZ22,CE22} imply FPRASes running in time $\widetilde{O}((n/\varepsilon)^2)$ as long as $\lambda<\lambda_c(\Delta)$. 
Yet much earlier, Weitz \cite{Wei06} gave the first fully polynomial-time approximation scheme (FPTAS, the deterministic counterpart to FPRAS) for the partition function of the hard-core model when $\lambda<\lambda_c(\Delta)$, which is not based on Markov chains.
While Weitz's algorithm has a running time $n^{O(\log \Delta)}$ in general,
it has an interesting feature that it gets faster as $\lambda$ decreases. 
Roughly speaking, for $k>0$ and $\lambda=O((1/\Delta)^{1+k})$,
Weitz's FPTAS runs in time $O(n^{1+1/k}/\eps^2)$.
In particular, if $\lambda=o(\Delta^{-2})$, Weitz's algorithm passes the $\Omega(n^2)$ barrier,
whereas the aforementioned MCMC method still takes $\Omega(n^2)$ time.
This leads to an intriguing question:
\begin{align}\label{question}
  \text{\emph{When can we achieve sub-quadratic running time for approximate counting?}}
\end{align}

In this paper we make some progress towards this question.
For hard-core models, Weitz's algorithm uses the self-reduction \cite{JVV86} to reduce approximate counting to estimating marginal probabilities (probabilities of a partial configuration of the system).
We provide a quadratic speedup for the marginal estimation step for $\lambda$ well below $\lambda_c(\Delta)$, albeit with the introduction of randomness.
The result is summarised as follows.

\begin{theorem}
\label{small-lambda}
Fix a constant $k>0$. 
Let $\Delta\ge 2$ be an integer and $\lambda<\frac{1}{\Delta^k(\Delta-1)}$. 
For graphs with maximum degree $\Delta$, there exists an FPRAS for the partition function of the hard-core model with parameter $\lambda$ in time $\widetilde{O}\big( \big( \frac{n}{\eps} \big)^{1+\frac{1}{2k}} \big)$, 
where $n$ is the number of vertices and $\eps$ is the error margin. 
\end{theorem}

\begin{remark}[Decay rate vs.~neighbourhood growth]
For a constant $\eps$,
the running time of \Cref{small-lambda} is sub-quadratic if $\lambda=o(\Delta^{-1.5})$,
and $\widetilde{O}(n^{1.5})$ if $\lambda=O(\Delta^{-2})$. 
In contrast, to achieve sub-quadratic running-time, Weitz's algorithm requires $\lambda=o(\Delta^{-2})$,
which is also the threshold when correlation decays faster than the growth of the neighbourhood.
This threshold has algorithmic significance in other contexts \cite{FGY22,AJ22},
but \Cref{small-lambda} implies that it is not essential to achieve sub-quadratic approximate counting.
\end{remark}

\Cref{fig:run-time} is a sketch comparing the running times of MCMC,\footnote{The running time of MCMC usually also depends on the parameter $\lambda$, but changing $\lambda$ does not change the exponent of $n$. The effect of $\lambda$ is usually a small polynomial factor hidden in the $\widetilde{O}(\cdot)$ notation, and the sketch in \Cref{fig:run-time} ignores this effect.} Weitz's algorithm, and \Cref{small-lambda}.\footnote{Another notable FPTAS is via zeros of polynomials \cite{Bar16,PR17}. It can achieve similar subquadratic running time when $\lambda=o(\Delta^{-2})$, but it is apparently no faster than Weitz's correlation decay algorithm.}
For the limiting case of $k=0$, our algorithm works when $\lambda<\frac{1}{\Delta-1}$ and still presents a quadratic speedup comparing to Weitz's algorithm.
However in this case the running time is $\left( \frac{n}{\eps} \right)^{O(\log \Delta)}$ and thus our speedup is hidden in the big-O notation and is less significant.
The parameter constraint $\frac{1}{\Delta-1}$ is imposed by the running-time tail bound of a subroutine we used, namely the recursive marginal sampler of Anand and Jerrum \cite{AJ22}.

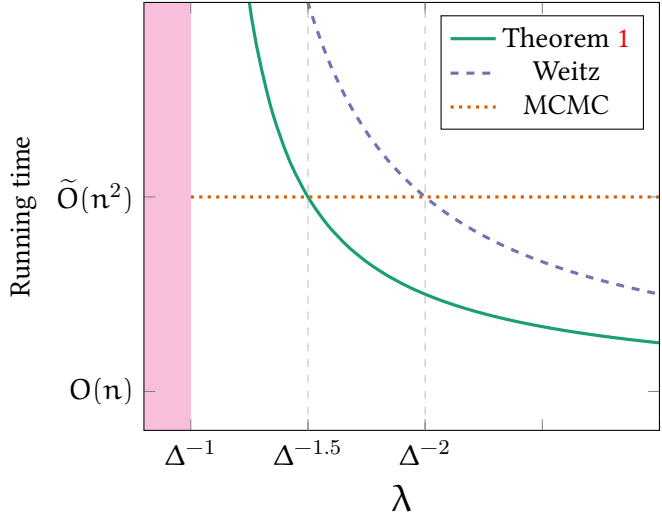
\begin{figure}[!htbp]
  \centering
  \begin{tikzpicture}[transform shape]
    \definecolor{myviolet}{RGB}{117,112,179}
    \definecolor{myorange}{RGB}{217,95,2}
    \definecolor{mygreen}{RGB}{27,158,119}
    \definecolor{mypink}{RGB}{231,41,138}
    
    \begin{axis}[width=0.54\linewidth,	xmin=0.8,   xmax=3, ymin=0.8,   ymax=3, xtick={1,1.5,2,2.5,3}, xticklabels={$\Delta^{-1}$,$\Delta^{-1.5}$,$\Delta^{-2}$, ,}, ytick={1,2}, yticklabels={$O(n)$,$\widetilde{O}(n^2)$}, extra x ticks={1.5,2}, extra x tick labels={\empty}, extra y ticks={}, extra tick style={dashed, grid=major},xlabel={\Large $\lambda$},ylabel={Running time},legend entries={\Cref{small-lambda},Weitz,MCMC}, axis on top, legend pos=north east]
    \addplot [very thick, draw=mygreen,  domain=1:3, smooth] {1+1/(2*(x-1))};
    \addplot [dashed, very thick, draw=myviolet,  domain=1:3, smooth] {1+1/(x-1)}; 
    \addplot [dotted, very thick, draw=myorange,  domain=1:3, smooth] {2};
    \addplot [fill=mypink!30,draw=none] coordinates {(0.8,0.8) (1,0.8) (1,3) (0.8,3)};
  \end{axis}
  \end{tikzpicture}
  \caption{Running time comparison among MCMC, Weitz's algorithm, and \Cref{small-lambda}}
  \label{fig:run-time}
\end{figure}

The key to our method is to find a new estimator of the marginal probability that simultaneously has low variance and can be evaluated very fast.
Our technique combines Weitz's self-avoiding walk (SAW) tree construction and the marginal sampler of Anand and Jerrum \cite{AJ22}.
The marginal of the root of the SAW tree preserves the desired marginal probability,
and can be evaluated in time linear in the size of the tree via standard recursion.
We use the unbiased marginal sampler to draw a random boundary condition at a suitable depth on the SAW tree, and compute the marginal of the root using recursion under this boundary condition.
Both steps can be computed in time near-linear in the size of the sub-tree.
In the self-reduction \cite{JVV86}, it suffices to have $O(1/n)$ variance for each marginal estimation.
Thus we only need absolute error $O(1/\sqrt{n})$ rather than the $O(1/n)$ error typically required by Weitz's algorithm.
This larger error tolerance roughly halves the depth where we truncate the SAW tree comparing to Weitz's truncation. 
As a result, we obtain our quadratic improvement on the marginal estimation step over Weitz's algorithm.
This method also extends to other anti-ferromagnetic 2-spin systems.

Our second contribution is about graphs with polynomial growth.
In particular, for planar graphs with quadratic growth,
we provide $\widetilde{O}((n/\eps)^{2-c})$ algorithms for some constant $c>0$.
An informal statement is as follows.
(The detailed statement is \Cref{lattice-counting}.)

\begin{theorem}\label{thm:main2}
  Let $\=G$ be a family of planar graphs with quadratic growth.
  For a spin system exhibiting spatial mixing on $\=G$,
  there exists an FPRAS for the partition function of $G\in\=G$ with $n$ vertices.
  The run-time is $\widetilde{O}((n/\eps)^{2-c})$ for some constant $c>0$,
  where $\eps$ is the error margin.
\end{theorem}

We note that one of the most important graphs in statistical physics, the 2D integer grid $\=Z^2$, indeed has quadratic growth.
More generally, any planar graph with a bounded radius circle packing has quadratic growth.
Thus \Cref{thm:main2} covers many important families of planar graphs, including most lattices.
(A \emph{non}-example would be the Cayley tree.)
Specialised to the hard-core model, 
\Cref{thm:main2} works up to the critical threshold, which is at least $\lambda_c(\Delta)$,\footnote{For a given graph family, such as subgraphs of $\mathbb{Z}^2$, the critical threshold may be well above $\lambda_c(\Delta)$.}
when the graph satisfies the condition in the theorem and has maximum degree $\Delta$.

The key to \Cref{thm:main2} is once again a suitable estimator for marginal probabilities.
We choose a distance~$\ell$ boundary around a vertex $v$ in $G$ with a carefully chosen $\ell$,
and our estimator is the marginal under random boundary conditions.
This boundary condition is yet again sampled using the algorithm of Anand and Jerrum \cite{AJ22}.
Our main observation is that due to quadratic growth, 
the number of possible boundary conditions do not grow very fast.
It turns out to be more efficient to create a look-up table by enumerating all boundary conditions first,
and instead of computing the marginal for each sample,
we simply find it in this table.
Since planar graphs have linear local tree-width,
the table can be created efficiently.
This last step is inspired by the work of Yin and Zhang \cite{YZ13}.

This method extends to any (not necessarily planar) graph families with polynomial growth.
Without planarity, we use brute-force enumeration instead to create the table.
This makes our gain on the running time smaller.
Again an informal statement is as follows, with the full version in \Cref{thm:polylog}.

\begin{theorem} \label{thm:main3}
  Let $\=G$ be a family of graphs with polynomial growth.
  For a spin system exhibiting spatial mixing on $\=G$,
  there exists an FPRAS for the partition function of $G\in\=G$ with $n$ vertices.
  The run-time is $\widetilde{O}\left(\frac{n^2}{\eps^2 2^{c(\log \frac{n}{\eps})^{1/d}}}\right)$ where $c>0$ is some constant, $d$ is the exponent of the polynomial growth,
  and $\eps$ is the error margin.
\end{theorem}

An example of such graphs would be the $d$-dimensional integer lattice $\=Z^d$.
Note that \Cref{thm:main2} is better than \Cref{thm:main3} for $d=2$ but requires the extra assumption of planarity.
The speedup factor $2^{c(\log n)^{1/d}}$ in \Cref{thm:main3} is slower than any polynomial in $n$ but faster than any polynomial in $\log n$.

We note an interesting related work by Chu, Gao, Peng, Sachdeva, Sawlani, and Wang \cite{CGPSSW18},
who give an approximate counting algorithm with running time $\widetilde{O}(m^{1+o(1)}+n^{1.875+o(1)}/\eps^{1.75})$ for spanning trees of graphs,
where $m$ is the number of edges and $n$ is the number of vertices. Notice that the input size here is $O(m)$ and $m=\Omega(n)$.
Thus their running time is also sub-quadratic.
However, there are some key differences between this work and ours.
Aside from not being a spin system, spanning trees can be counted exactly in polynomial time, 
thanks to Kirchhoff's matrix-tree theorem.
This allows them to use various efficient exact counting subroutines,
whereas the problems we consider are $\numP$-hard in general and no such subroutine is likely to exist.

Another more recent related result is the sub-quadratic all-terminal unreliability estimation algorithm by Cen, He, Li, and Panigrahi \cite{CHLP23},
which runs in sub-quadratic time $m^{1+o(1)}\eps^{-3}+\widetilde{O}(n^{1.5}\eps^{-2})$.
This problem, while $\numP$-hard, is not a spin system either.
Their method features a recursive Monte Carlo estimator that is very different from ours,
and not applicable to spin systems.

A crucial ingredient of our algorithm is the recursive marginal sampler of Anand and Jerrum \cite{AJ22}, which allows us to sample configurations on part of the graph while examining a low number of vertices in expectation.
This type of local / marginal sampler enables partial access to a large random object (in this case, a random spin configuration on the whole graph) with substantially less information than traditional samplers.
It has found applications in local computation algorithms \cite{BRY20}, and in derandomising Markov chains \cite{FGWWY22}.
Our results offer yet another application, namely to accelerate computation of the global partition function. 

We hope that our results are just the first step towards answering Question \eqref{question}.
In particular, it is not clear whether an $O((n/\eps)^{2-c})$ algorithm exists for the hard-core model when $\lambda=\Theta(1/\Delta)$ on graphs with maximum degree $\Delta$,
or if more efficient algorithms exist for graphs with polynomial or sub-exponential growth.
We leave these questions as open problems.

\section{Preliminaries}\label{sec:prelim}
We are interested in spin systems which exhibit strong spatial mixing. 

\begin{definition}
  A $q$ state spin system (or $q$-spin system for short) is given by a graph $G=(V,E)$, a $q$-by-$q$ interaction matrix $A$, and a field $b:[q]\rightarrow\mathbb{R}$. A configuration of $G$ is an assignment of states to vertices, $\sigma: V\rightarrow[q]$. The weight of a configuration $\sigma$ is determined by the assignments to the vertices and the interactions between them,
\[w(\sigma)\defeq\prod_{(u,v)\in E}A_{\sigma(u),\sigma(v)}\prod_{v\in V}b_{\sigma(v)}.\]

The Gibbs distribution $\mu$ is one where the probability of each configuration is proportional to its weight, namely, $\mu(\sigma)\defeq\frac{w(\sigma)}{Z(G)}$, where the partition function $Z(G)=\sum_\sigma w(\sigma)$ is a normalising factor. 
\end{definition}

In this paper, we consider the following permissive spin system, which says any locally feasible configuration can be extended to a globally feasible configuration.
\begin{definition}
  A $q$-spin system on $G=(V,E)$ is permissive if for any $\Lambda \subseteq V$, any $\sigma \in [q]^\Lambda$, if $b_{\sigma(v)} > 0$ for all $v \in \Lambda$ and $A_{\sigma(u),\sigma(v)} > 0$ for all $u,v \in \Lambda$ satisfying $(u,v) \in E$, then $\sigma$ can be extended to a full configuration $\sigma' \in [q]^V$ such that $w(\sigma')>0$.
\end{definition}
Many natural spin systems are permissive. Examples include the hard-core model, the graph $q$-colouring with $q \geq \Delta + 1$, where $\Delta$ is the maximum degree of the graph, and all spin systems with soft constraints (e.g. the Ising model and the Potts model).

We call the problem of evaluating $Z$ the counting problem for the $q$-spin system.
The standard algorithmic aim here is a \emph{fully-polynomial randomised approximation scheme} (FPRAS),
where given the spin system and an accuracy $\eps>0$,
the algorithm outputs $\widetilde{Z}$ such that $1-\eps\le\frac{\widetilde{Z}}{Z}\le 1+\eps$ with probability at least $3/4$, and runs in time polynomial in the size of the system and $1/\eps$.
To understand the requirement of an FPRAS, note that the probability $3/4$ can be boosted arbitrarily close to $1$ via standard means.
The accuracy can also be boosted by taking many disjoint copies of the system.
In fact, any polynomial accuracy can be boosted to an arbitrarily small $\eps$ in polynomial-time.

Also note that if $G$ is disconnected, then $Z(G)=\prod_{i}Z(G_i)$ where $G_i's$ are the connected components of $G$.
Thus, we always consider connected graphs in the paper.

Similar to $\mu(\sigma)$ for the probability of a configuration, for an $S\subseteq V$ and a partial configuration $\sigma_S$ on $S$, we use $\mu(\sigma_S)$ for the marginal probability of $\sigma_S$ under $\mu$.
We denote the marginal distribution induced by $\mu$ on $S$ by $\mu_S$.
When $S=\{v\}$, we also write $\mu_v$. For the distribution conditioned on a partial configuration $\sigma_S$, we use $\mu^{\sigma_S}$ or $\mu_v^{\sigma_S}$. 

Strong spatial mixing is a property of the spin system where a partial configuration of $G$ does not significantly influence the assignment of a distant vertex. 

\begin{definition}[SSM]\label{def:SSM}
A $q$-spin system is said to have strong spatial mixing with decay rate $f(\ell)$ for a graph $G=(V,E)$ if for any $v\in V, S\subset V$, and two configurations $\sigma_S,\tau_S$,
\[\dTV(\mu_v^{\sigma_S},\mu_v^{\tau_S})\leq f(\ell),\]
where $\dTV$ denotes the total variation distance, $T\subseteq S$ is the subset where the configurations are different, and $\ell=\dist(v,T)$ is the minimum distance from $v$ to a vertex in $T$.

We say a $q$-spin system has strong spatial mixing for a family of graphs $\mathbb{G}$ if there exists $f(\cdot)$ such that the system has strong spatial mixing with the same decay rate $f(\ell)$ for all $G\in\mathbb{G}$. 
\end{definition}

Strong spatial mixing is a very strong form of correlation decay. When $f(\ell)=\exp(-\Omega(\ell))$ we say we have strong spatial mixing with exponential decay. 

\subsection{Two-state spin systems}

A spin system is symmetric if $A_{ij}=A_{ji}$ for all $i,j$. 
When $q=2$ and the system is symmetric, we have states $\{0,1\}$ and can normalise $A$ and $b$ so that the interaction between $0$ and $1$ and the contribution of $0$ are $1$, and $A=
\begin{bmatrix}
\beta & 1\\
1 & \gamma
\end{bmatrix}$ and $b=(1,\lambda)$ for $\beta,\gamma\geq0$, and $\lambda>0$. 

When $\beta=\gamma$ the system is an Ising model, and for $\beta=1, \gamma=0$ the system is a hard-core gas model. We call a system \emph{anti-ferromagnetic} if disagreeing assignments of adjacent vertices are more heavily weighted, namely $\beta\gamma<1$. 

For a tree $T$ rooted at $v$ and a partial configuration $\sigma_S$ we define the marginal ratio
\begin{align*}
  R_T^{\sigma_S}\defeq\frac{\mu_v^{\sigma_S}(1)}{\mu_v^{\sigma_S}(0)}=\frac{\mu_v^{\sigma_S}(1)}{1-\mu_v^{\sigma_S}(1)},
\end{align*}
or $R_T^{\sigma_S}\defeq\infty$ if $\mu_v^{\sigma_S}(1)=1$.
These ratios satisfy a well-known recurrence relation:
\begin{align}
\label{saw-recursion}
R_T^{\sigma_S}=\lambda\prod_{i=1}^d\frac{\gamma R_{T_i}^{\sigma_S}+1}{R_{T_i}^{\sigma_S}+\beta},
\end{align}
where $T_i$ is the $i$th subtree of $T$. Similarly, for a graph $G$ we can define $R_{G,v}^{\sigma_S}=\mu_v^{\sigma_S}(1)/(1-\mu_v^{\sigma_S}(1))$.
While $R_{G,v}^{\sigma_S}$ does not admit a simple recursion, the self-avoiding walk (SAW) tree of $G$ at $v$ as constructed by Weitz \cite{Wei06} can be used to compute it. 

\begin{theorem}[Theorem 3.1 of \cite{Wei06}]
\label{saw-theorem}
For any $G=(V,E)$, a configuration $\sigma_S$ on $S\subset V$, and any $v\in V$, there exists a tree $\TSAW=\TSAW(G,v)$ such that
\[R_{G,v}^{\sigma_S}=R_{\TSAW}^{\sigma_S}.\]
\end{theorem}

The SAW tree is rooted at $v$.
Each node corresponds to a self-avoiding walk starting from $v$. The length of the walk is the same as the distance between the node and the root $v$.
When a walk is closed, the node is set to unoccupied or occupied according to if the penultimate vertex is before or after the starting vertex of the cycle in some pre-determined local ordering at the last vertex.
For details, see \cite{Wei06}.

The SAW tree can have depth up to $n$, so may be exponential in size. Marginals on the SAW tree are therefore difficult to compute, but using the recursion in \Cref{saw-recursion} we can approximate them by truncating the tree. 
This approximation is accurate when strong spatial mixing holds,
and the time to compute the marginal is linear in the size of the truncated tree.
To maintain a polynomial running time, Weitz~\cite{Wei06} choose to truncate it at a suitable logarithmic depth.


\section{Fast SSM regime for 2-spin systems}
In this section we give a quadratic speedup of Weitz's Algorithm to estimate the marginal of a single vertex in $2$-spin systems, albeit being randomised instead of deterministic.
We use the hard-core model as our running example to illustrate the main ideas.
The main result of the section is \Cref{small-lambda}.


Let the hard-core model be described by $A=\begin{bmatrix}1&1\\1&0\end{bmatrix}$ and $b=(1,\lambda)$.
The support of the Gibbs distribution is the set of independent sets of $G$. 
Let vertices assigned $0$ not be in the independent set (unoccupied) and vertices assigned $1$ be in the independent set (occupied). 
Our algorithm uses self-reduction \cite{JVV86} as follows.
Since unoccupied vertices contribute $1$ to the weight of a configuration, we can consider the all $0$ configuration $\sigma_0$ where 
\begin{align}
\label{hardcore-reduction}
\frac{1}{Z(G)}=\mu(\sigma_0)&=\mu_{v_1}(0)\mu_{V\setminus\{v_1\}}^{v_1\leftarrow0}(\mathbf{0})
=\mu_{v_1}(0)\frac{1}{Z(G\setminus\{v_1\})}
=\mu_{G_1,v_1}(0)\mu_{G_2,v_2}(0)\cdots\mu_{G_n,v_n}(0),
\end{align}
and where $G_i=G\setminus\{v_1,\ldots,v_{i-1}\}$ for all $i\in[n]$.
This reduces the problem of computing $Z(G)$ to computing $\mu_{v_1}$ and recursively $Z(G\setminus\{v_1\})$.
As the $G_i$'s are subgraphs of $G$, they have the same degree bound and still exhibit SSM.
The crux of our algorithm is to design a random variable that estimates $\mu_{v}$ in time $\widetilde{O}(n^{1/(2k)})$.



Another ingredient we need is the lazy single-site sampler by Anand and Jerrum \cite{AJ22},
which allows us to rapidly sample a partial configuration vertex by vertex. 
The original setting of \cite{AJ22} requires sub-exponential neighbourhood growth in order to work up to the strong spatial mixing threshold, but in our parameter regime no sub-exponential growth is required. 
Moreover, only the expected running time is studied in \cite{AJ22}, while we need a tail bound. 
A similar analysis is done in \cite[Appendix B]{FGWWY22}. 
For completeness, we provide a proof specialised to our setting in \Cref{sec:truncate-AJ}.

\begin{lemma}  \label{lem:truncate-AJ}
  Let $\Delta\ge 2$ be an integer and $\lambda<\frac{1}{\Delta-1}$.
  Let $G=(V,E)$ be a graph with maximum degree $\Delta$.
  There exists an algorithm that, 
  for any $v\in V$,
  draws a sample from $\mu_v$
  and halts in time $O(\log \frac{1}{\eps})$ with probability at least $1-\eps$.
\end{lemma}


Our algorithm then combines the lazy sampler of \Cref{lem:truncate-AJ} with the SAW tree of \cite{Wei06}.
We expand the SAW tree, and then use \Cref{lem:truncate-AJ} to sample a truncated boundary,
from which we use the recursion in \eqref{saw-recursion} to get our estimate.
The depth of the truncation controls the variance of this estimator.
In our algorithm, we only need to bound the variance from above by $1/n$.
In contrast, Weitz's algorithm requires the error of the marginal incurred by the truncation to be bounded from above by $O(1/n)$.
As the variance of our estimator decays twice as fast as the marginal errors,
our truncation depth is roughly half of that in Weitz's algorithm.
Consequently, we achieved a quadratic speedup for estimating each term in \eqref{hardcore-reduction}.

\begin{lemma}
\label{small-lambda-alg}
Let $\=G_{\Delta}$ be the family of graphs with maximum degree $\Delta$.
Suppose the hard-core model has strong spatial mixing with decay rate $C\Delta^{-k\ell}$ for some constant $C>0$ for $\=G$.
Then there exists an algorithm that, for any $G=(V,E)\in\=G$ and $v\in V$, generates a random sample $\widetilde{p}_v$ and halts in time $O(n^{1/(2k)}(\log \frac{n}{\delta})^2)$ with probability at least $1-\frac{\delta}{8}$.
Furthermore, $\Ex[\widetilde{p}_v]=\mu_v(0)$ and $\Var{}{\widetilde{p}_v}\le 1/n$. 
\end{lemma}

\begin{proof}
Let $\TSAW$ be the self-avoiding walk tree for $G$ rooted at $v$ as defined in \Cref{saw-theorem}, and let $S=\{u\in V|d_{\TSAW}(v,u)=\ell\}$ where $\ell$ is a parameter we will fix later. We have
\[\mu_{\TSAW,v}(0)=\sum_{\sigma\in\{0,1\}^S}\mu_{\TSAW}(\sigma)\mu_{\TSAW,v}^\sigma(0)=\Ex_{\sigma\sim\mu_{\TSAW,S}}[\mu_{\TSAW,v}^\sigma(0)].\]
We use \Cref{lem:truncate-AJ} to sample $\sigma$.
Fix an arbitrary order of $S=\{s_1,s_2,\ldots,s_{\abs{S}}\}$.
We sample first the marginal of $s_1$ with $\eps\defeq\frac{\delta}{8\abs{S}}$.
Then, conditioned on the result on $s_1$, we sample $s_2$ with the same $\eps$, and so on and so forth.
Note that whatever the result on $s_1$ is, it always reduces to a hard-core instance of a smaller graph.
Thus, the condition of \Cref{lem:truncate-AJ} is always satisfied until all of $S$ are sampled.
This gives a boundary condition $\sigma_S$ in $\TSAW$.

As the full SAW tree may be exponential in size,
a little care is required to implement the outline above.
We first expand $\TSAW$ up to level $\ell$, denoted $T_\textrm{SAW,$\ell$}$. 
The algorithm in \Cref{lem:truncate-AJ} (\Cref{Alg:ssms-hardcore} in \Cref{sec:truncate-AJ}) is essentially an exploration process.
When we apply it to sample the boundary condition~$\sigma_S$,
we expand the SAW tree below $T_\textrm{SAW,$\ell$}$ on the fly,
only creating vertices that are explored by the algorithm.
Note that the construction of the SAW tree imposes a boundary condition whenever a vertex in $G$ is encountered again in a self-avoiding walk.
We implement this pinning by remembering a list of all ancestors of a given node in the SAW tree and checking the next vertex to explore against this list.
Since \Cref{lem:truncate-AJ} halts in $O(\log\frac{1}{\eps})$ time with probability at least $1-\eps$,
this extra check incurs a multiplicative slowdown factor $O(\ell+\log\frac{1}{\eps})=O(\ell+\log\frac{\abs{S}}{\delta})$ with probability at least $1-\eps$. 

Given $\sigma_S$, we can compute $\mu_v^{\sigma_S}(0)=\widetilde{p}_v$ with the standard dynamic programming approach.
By a union bound, the total running time of sampling the boundary is $O(\abs{S}\log\frac{\abs{S}}{\delta}(\ell+\log\frac{\abs{S}}{\delta}))$ with probability at least $1-\frac{\delta}{8}$,
  and the dynamic programming step uses time $O(\abs{T_\textrm{SAW,$\ell$}})$.

We choose $\ell\defeq\left\lceil\frac{\log(n)/2-\log C}{k\log(\Delta)}\right\rceil$ so that $C\Delta^{-k\ell}\le n^{-0.5}$ and $\Delta^\ell\le C' n^{1/(2k)}$ for some constant $C'>0$. 
Note that $|S|\le(\Delta-1)^\ell$ and $|T_\textrm{SAW,$\ell$}|\le \Delta^\ell$. 
Then the total runtime to draw a sample is $O(n^{1/(2k)}(\log \frac{n}{\delta})^2)$ with probability at least $1-\frac{\delta}{8}$.

Finally, we analyse the variance. Strong spatial mixing implies that $\abs{\mu_v^{\sigma_S}(0)-\mu_v(0)}\leq C\Delta^{-k\ell}$ for any $\sigma_S$, so
\begin{align}  \label{eqn:var-error}
  \Var{}{\widetilde{p}_v}=\Var{\sigma_S}{\mu_v^{\sigma_S}(0)}=\Ex_{\sigma_S\sim\mu_{\TSAW,S}}[\abs{\mu_v^{\sigma_S}(0)-\mu_v(0)}^2]
  &\leq \left(C\Delta^{-k\ell}\right)^2\leq n^{-1},
\end{align}
which is what we desire.
\end{proof}

\begin{lemma}
\label{lem:contraction}
  For a graph $G$ with maximum degree $\Delta$, if $\lambda\leq\frac{1}{\Delta^k(\Delta-1)}$ for some constant $k>0$, the hard-core model on $G$ exhibits strong spatial mixing with decay rate $C\Delta^{-k\ell}$. 
\end{lemma}

\begin{proof}
  It is well-known that if $\lambda<\lambda_c(\Delta) = \frac{(\Delta-1)^{\Delta-1}}{(\Delta-2)^{\Delta}}\approx\frac{e}{\Delta}$,
  strong spatial mixing holds with exponential decay $Cr^{\ell}$ for some constant $C$ and $r<1$ \cite{Wei06}.
  Moreover, the decay rate $r$ can be controlled by a quantity related to the recursion \eqref{saw-recursion} \cite{SST14}.
  For example, by \cite[Lemma 7.20]{Guo15}, $r$ is bounded by $r\le|f'(\widehat{x})|$,
  where $f(x)\defeq\frac{\lambda}{(1+x)^{\Delta-1}}$ is the symmetric version of the recursion in \eqref{saw-recursion} and $\widehat{x}$ is the unique positive fixed point of $f$.
  (Note that when the degree of $G$ is at most $\Delta$, all vertices but the root in $\TSAW$ have branching number $\Delta-1$.)
  Then we have
  \begin{align*}
    \abs{f'(x)} = \abs{-\frac{f(x)(\Delta-1)}{1+x}} < (\Delta-1) f(x).
  \end{align*}
  As $\widehat{x}>0$ and $\widehat{x}$ is a fixed point,
  \[\widehat{x}<\widehat{x}(1+\widehat{x})^{\Delta-1}=\lambda.\]
  Thus, as $\lambda\leq\frac{1}{\Delta^k(\Delta-1)}$,
  \begin{align*}
    r&\le|f'(\widehat{x})|< (\Delta-1) f(\widehat{x})=(\Delta-1)\widehat{x}<\frac{1}{\Delta^k}. \qedhere
  \end{align*}
\end{proof}

Now we are ready to prove \Cref{small-lambda}.

\def\epsnot{\eps_0}

\begin{proof}[Proof of Theorem~\ref{small-lambda}]
  We first give an algorithm that runs in time $O\left(\frac{n^{1+1/(2k)}}{\eps^2}(\log \frac{n}{\eps})^2\right)$, 
  and then improve the dependency on $\eps$ by a simple trick at the end of the proof.
  It would be easier to first introduce an algorithm that runs within the desired bound with high probability.
  To have a fixed running time upper bound, we then truncate the algorithm.

  Set $N\defeq\lceil 8e^{(1+\lambda)^2}/\epsnot^2 \rceil$ where $\epsnot=\eps/2$.
  Let $X\defeq\prod_{i=1}^n \widetilde{p}_{G_i,v_i}$ where $G_1=G$ and $G_i=G_{i-1}\setminus\{v_{i-1}\}$.
  By \Cref{lem:contraction}, we can use \Cref{small-lambda-alg} to draw $N$ samples of $X$ and take its average, where we set $\delta=\frac{1}{nN}$ in  \Cref{small-lambda-alg}.
  Each $\widetilde{p}_{G_i,v_i}$ can be computed in time $O(n^{1/(2k)}(\log \frac{n}{\delta})^2)$ with probability at least $1-\frac{\delta}{8}$, 
  so computing one sample of $X$ takes time $O(n^{1+1/(2k)}(\log \frac{n}{\delta})^2)$ time with probability at least $1-\frac{n\delta}{8}$ by a union bound. 
  By a union bound again, the overall running time of taking the average is $O(Nn^{1+1/(2k)}(\log \frac{n}{\delta})^2)=O\left(\frac{n^{1+1/(2k)}}{\eps^2}(\log \frac{n}{\eps})^2\right)$ 
  with probability at least $1-\frac{\delta}{8}\cdot nN = \frac{7}{8}$.
  

  Since $\{\widetilde{p}_{G_i,v_i}\}$ are mutually independent, by \Cref{small-lambda-alg},
  \begin{align*}
    \Ex[X]=\Ex\left[\prod_{i=1}^n\widetilde{p}_{G_i,v_i}\right] & = \prod_{i=1}^n\mu_{G_i,v_i}(0)=\frac{1}{Z(G)}.
  \end{align*}
  We bound $\Var{}{X}$ as follows
  \begin{align*}
    \frac{\Var{}{X}}{(\Ex[X])^2}&=\frac{\Ex[X^2]}{(\Ex[X])^2}-1 = \frac{\prod_{i=1}^n\Ex[\widetilde{p}_{G_i,v_i}^2]}{\prod_{i=1}^n\Ex[\widetilde{p}_{G_i,v_i}]^2}-1\\
    &=\prod_{i=1}^n\left(1+\frac{\Var{}{\widetilde{p}_{G_i,v_i}}}{(\Ex[\widetilde{p}_{G_i,v_i}])^2}\right)-1\\
    &\leq\left(1+\frac{c}{n}\right)^n-1\tag*{(by \Cref{small-lambda-alg})}\\
    &<e^{c},
  \end{align*}
  where $c=\max_i(1/\mu_{G_i,v_i}(0)^2)$. 
  Note that as $\mu_{G_i,v_i}(0)\ge \frac{1}{1+\lambda}$, $c\le (1+\lambda)^2$ and is a constant.

  Let $\widetilde{X}$ be the average of $N$ samples of $X$.
  Then $\Var{}{\widetilde{X}} = \frac{\Var{}{X}}{N} \le \frac{e^{c}}{N\cdot Z(G)^2}$.
  By Chebyshev's inequality,
  \begin{align*}
    \Pr\left[\abs{\widetilde{X}-\frac{1}{Z(G)}}\ge\frac{\epsnot}{Z(G)}\right]\le\frac{\Var{}{\widetilde{X}}}{\frac{\epsnot^2}{Z(G)^2}} 
    \le \frac{e^{c}}{N\cdot Z(G)^2} \cdot \frac{Z(G)^2}{\epsnot^2}\le \frac{1}{8}.
  \end{align*}
  Thus, with probability at least $7/8$,
  we have that $\frac{1-\epsnot}{Z(G)}\le \widetilde{X} \le \frac{1+\epsnot}{Z(G)}$.
  Finally, we output $\widetilde{Z}=1/\widetilde{X}$.
  To make sure that the algorithm runs within the time bound $O\left(\frac{n^{1+1/(2k)}}{\eps^2}(\log \frac{n}{\eps})^2\right)$, 
  we truncate the algorithm if it runs overtime and output an arbitrary value in that case.
  This truncated version can be coupled with the untruncated algorithm with probability at least $7/8$,
  and its output $\widetilde{Z}$ satisfies $1-\eps\le \frac{\widetilde{Z}}{Z(G)}\le 1+\eps$ with probability at least $7/8-1/8=3/4$.

  To improve the dependency on $1/\eps$,
  construct $G'$ as $t=\ceil{2/\eps}$ copies of $G$.
  Without loss of generality we may assume that $\eps<1$ so that $t>2$
  Denote $Z'=Z(G')$ and $Z=Z(G)$.
  Then $Z'=Z^m$.
  We set $\eps_1=1-1/e$ and apply the algorithm above on $G'$ to get an estimate $\widetilde{Z}$ to $Z'$,
  which satisfies
  \begin{align*}
    e^{-1}\le\frac{\widetilde{Z}}{Z'}\le 2-1/e < e.
  \end{align*}
  Thus,
  \begin{align*}
    e^{-1/t}\le \frac{\widetilde{Z}^{1/t}}{Z} \le e^{1/t}.
  \end{align*}
  As $t=\ceil{2/\eps}>2$, $e^{1/t}<1+2/t\le 1+ \eps$ and $e^{-1/t} > 1-\frac{1}{t}>1-\eps$. 
  Thus, $\widetilde{Z}^{1/t}$ is an estimate to $Z$ within the desired accuracy.
  As for the running time of this algorithm, 
  the size of $G'$ is $t n=O\big(\frac{n}{\eps}\big)$ and $\eps_1=\Omega(1)$.
  Thus, the overall running time is $O\big( \big( \frac{n}{\eps} \big)^{1+\frac{1}{2k}}(\log \frac{n}{\eps})^2 \big)=\widetilde{O}\big( \big( \frac{n}{\eps} \big)^{1+\frac{1}{2k}} \big)$.
\end{proof}

Note that, Weitz's algorithm is faster if the correlation decay is faster,
but in that case so is our algorithm.
In \Cref{sec:lb-Weitz},
\Cref{lemma-inf-lower} shows that the correlation decay cannot be much faster than the standard analysis in the parameter regimes of \Cref{small-lambda},
and our speed-up, comparing to Weitz's algorithm, is always at least $\widetilde{O}(n^{1/2k-o(1/k^2)})$.

%
%

We also remark that \Cref{small-lambda} generalises to antiferromagnetic 2-spin systems.
This is because all the key ingredients, namely correlation decay, Weitz's SAW tree, and the marginal sampler of Anand and Jerrum all generalise,
except that the Anand-Jerrum algorithm would require the neighbourhood growth rate smaller than the decay rate (see \Cref{generic-AJ}).
This is also the parameter regime where Weitz's algorithm is faster than $O(n^2)$.
Thus, our speedup is still in the sub-quadratic regime.
The self-reduction in \eqref{hardcore-reduction} also generalises (as we will see in \eqref{q-spin-decomposition} in the next section).
One needs to redo the calculations in \Cref{lem:contraction} to get a precise statement, which we will omit here.

\section{Speed-up on planar graphs}
In this section we mainly consider (not necessarily two-state) spin systems on planar graphs.
We show that for any planar graph with quadratic neighbourhood growth,
when SSM holds with exponential decay,
approximate counting can be done in sub-quadratic time.
For example, this includes all subgraphs of the 2D lattice $\=Z^2$.
The circle-packing theorem asserts that any planar graph is the tangent graph of some circle packing.
All planar graphs with bounded-radius circle packings have quadratic neighbourhood growth.
Thus this is a substantial family of planar graphs.
Moreover, in \Cref{sec:poly-growth} we extend the result to (not necessarily planar) graphs with polynomial growth,
but the speed up factor there is sub-polynomial yet faster than $(\log n)^k$ for any $k$.

\begin{definition}  \label{def:quadratic-growth}
  A graph family $\=G$ has \emph{quadratic growth},
  if there is a constant $C_0$ such that for any $G=(V,E)\in\=G$, $v\in V$, and any integer $\ell>0$,
  $\abs{B_v(\ell)}\le C_0\ell^2$.
\end{definition}

Subgraphs of the 2D lattice $\=Z^2$ satisfies \Cref{def:quadratic-growth} with $C_0=5$.
Note that by taking $\ell=1$, \Cref{def:quadratic-growth} implies that the maximum degree is no larger than $C_0$.

\begin{theorem}\label{lattice-counting}
  Let $\=G$ be a family of planar graphs with quadratic growth (assume the rate is $C_0\ell^2$).
  Let ${\mathbf{A}}$ and ${\mathbf{b}}$ specify a $q$-state spin system, which exhibits SSM with decay rate $Cr^{-\ell}$ on $\=G$.
  Then there is a constant $c>0$ such that there exists an FPRAS for the partition function of $G\in\=G$ with $n$ vertices with run-time $\widetilde{O}\big(\big(\frac{n}{\eps}\big)^{2-c}\big)$.
  The constant $c$ depends on $C_0$, $q$, and~$r$.
\end{theorem}

\Cref{lattice-counting} is the detailed version of \Cref{thm:main2}.

Essentially the idea is still to find an estimator for the marginal of an arbitrary vertex that can be evaluated very quickly.
Let us first consider a $\sqrt{n}$-by-$\sqrt{n}$ grid.
For any vertex $v$, we consider the sphere $S_v(\ell)$ of radius $\ell=O(\log n)$ centered at $v$,
and a random configuration $\tau$ on $S_v(\ell)$.
Let $B_v(\ell)$ be the ball of radius $\ell$ centered at $v$.
Since any planar graph has linear local tree-width \cite{DH04,Epp00},
$B_v(\ell)$ has tree-width $O(\ell)$.
Thus, given a configuration $\tau$ on $S$, 
the law of $\mu_v^{\tau}$ can be computed in time $2^{O(\ell)}\textrm{poly}(\ell)$ for a fixed $\tau$ (see, e.g.~\cite{YZ13}\footnote{The algorithm in \cite{YZ13} uses the separator decomposition. Another possibility is to first find a constant approximation of the tree decomposition first \cite{KT16}, and then apply Courcelle's theorem.}).
This step can be very efficient with a carefully chosen $\ell$.

For a general bounded degree planar graph, $\abs{S_v(\ell)}$ can be a polynomial in $n$, which makes the number of possible $\tau$'s exponential in $n$.
However, for a $\sqrt{n}$-by-$\sqrt{n}$ grid, $\abs{S_v(\ell)}\le 4\ell=O(\log n)$,
and the number of possible $\tau$'s is much smaller and is a small polynomial in $n$.
Thus, it would be more efficient to first create a table to list all possibilities of $\tau$,
and then, instead of computing $\mu_v^{\tau}$ each time,
 simply look up the answer from this table. We can do the same for any subgraph of $\mathbb{Z}^2$ by choosing a boundary based on distance in the original grid. 

For a general $G\in\=G$,
we no longer have a linear bound on the size of the boundary.
See \Cref{sec:quadratic-boundary} for a subgraph of $\=Z^2$ where the distance $\ell$ boundary has size $\Omega(\ell^2)$.
However, since $\=G$ has quadratic growth,
we know that $B_v(\ell)\le C_0\ell^2$ for some constant $C_0>0$.
It implies that 
\begin{align*}
  \sum_{i=\ell/2}^{\ell}\abs{S_v(\ell)}\le\abs{B_v(\ell)} \le C_0\ell^2.
\end{align*}
Thus, there must exist an $\ell'\in[\ell/2,\ell]$ such that $\abs{S_v(\ell')}\le 2C_0\ell$.
We will find this $\ell'$ and use $S_v(\ell')$ instead.

Once again, we use a self-reduction similar to \eqref{hardcore-reduction}. For $q$-spin systems, given a feasible configuration $\sigma$, we have the decomposition,
\begin{align}
	\label{q-spin-decomposition}
	\frac{w(\sigma)}{Z_G}=\mu(\sigma)=\mu_{v_1}(\sigma_{v_1})\mu_{v_2}^{\sigma_{v_1}}(\sigma_{v_2})\mu_{v_3}^{\sigma_{v_1},\sigma_{v_2}}(\sigma_{v_3})\ldots\mu_{v_n}^{\sigma_{v_1},\ldots,\sigma_{v_{n-1}}}(\sigma_{v_n}).
\end{align}
When computing our table, we will have to condition on the already pinned vertices. 

\begin{lemma}\label{YZ-table}
  Let ${\mathbf{A}}$, ${\mathbf{b}}$, $q$ and $G$ be as in \Cref{lattice-counting}.
  For $v\in V$, a partial configuration $\sigma$, and an integer~$\ell$, 
  we can find an $\ell'$ such that $\ell'\in[\ell/2,\ell]$,
  and then construct a table of $\mu_v^{\sigma,\tau}$, indexed by every boundary configuration $\tau$ on unpinned vertices of $S_v(\ell')$.
  The total run-time is $2^{C_1\ell}$,
  where~$C_1$ is a constant depending on $C_0$ and $q$.
\end{lemma}

\begin{proof}
  As discussed earlier, due to the quadratic growth of $G$,
  there must exist an $\ell$ such that $\ell'\in[\ell/2,\ell]$ and $\abs{S_v(\ell')}\le 2C_0\ell$.
  To find this $\ell$,
  we do a breadth-first-search to check $S_v(i)$ from $i=\ell/2$ to $\ell$.
  The running time is at most $O(B_v(\ell))=O(\ell^2)$.

  Once $\ell'$ is found,
  $\abs{S^{\sigma}_v(\ell')}\le\abs{S_v(\ell')}\le 2C_0\ell$,
  and there are at most $q^{2C_0\ell}$ config\-urations~$\tau$ in our table.
  As $G$ is a planar graph, the tree-width of the ball $\tw(B_v(\ell'))=O(\ell')=O(\ell)$.
  Thus, using for example the algorithm of \cite{YZ13}, each entry of the table can be computed in time $2^{O(\ell)}\poly{\ell}$.
  The total amount of time required is $O(\ell^2)+q^{2C_0\ell}2^{O(\ell)}\poly{\ell}\le 2^{C_1\ell}$,
  for some sufficiently large constant~$C_1$.
\end{proof}

While we may construct this table very quickly,
it is not clear how to compute or estimate the marginals of the boundary condition $\tau$'s rapidly.
Instead, we sample a random one using the marginal sampler \cite{AJ22} that terminates in almost linear time with high probability. 
See \Cref{generic-AJ}. 

\begin{lemma}
	\label{lem:lattice-estimator}
    Let ${\mathbf{A}}$, ${\mathbf{b}}$, $q$ and $G\in\=G$ be as in \Cref{lattice-counting}.
    Let $\sigma$ be a partial configuration.
    For any $v\in V$ not pinned under $\sigma$ and any $k\in[q]$,
    there exists an algorithm that generates a random variable $\widetilde{Z}$ such that $\Ex[\widetilde{Z}]=\mu_{v}^{\sigma}(k)$ and $\Var{}{\widetilde{Z}}\leq1/n$.
    Moreover, its running time is $\widetilde{O}(n^{1-c})$ with high probability where $c$ depends on $C_0$, $q$, and~$r$.
\end{lemma}

\begin{proof}
  Let $\ell$ be a constant that we will choose later.
  Let $\ell'\in[\ell/2,\ell]$ be as in \Cref{YZ-table},
  and let $\tau$ be a boundary condition on the unpinned vertices of $S_{v}(\ell')$ under $\sigma$.
  Let $Z_{v}(\tau)=\mu_{v}^{\sigma,\tau}(k)$ so that $\Ex_\tau[Z_{v}(\tau)]=\mu_{v}^{\sigma}(k)$. 
  Then, let
  \[\widetilde{Z}\defeq\frac{1}{m}\sum_{j=1}^mZ_{v}(\tau_j)\]
  be the empirical mean over $m$ random samples $\tau_j$,
  where we will choose $m$ later. 

  Since the spin system exhibits SSM with decay rate $Cr^{-\ell}$,
  similar to \eqref{eqn:var-error},
  it follows that $\Var{}{Z_{v}(\tau)}\le C^2r^{-2\ell'}\le C^2r^{-\ell}$. Then
  \[\Var{}{\widetilde{Z}}=\Var{}{\frac{1}{m}\sum_{j=1}^mZ_{v}(\tau_j)}\le\frac{C^2}{mr^{\ell'}}.\]
  Thus, we set $m=\lceil nC^2r^{-\ell'} \rceil$ samples so that $\Var{}{\widetilde{Z}}\leq 1/n$. 

  For the running time,
  we first construct the table as in \Cref{YZ-table}.
  Then we take $m$ samples of $\tau$, each of which can be generated in time almost linear in $\abs{S_{v}(\ell')}$ with high probability using \Cref{generic-AJ}.
  As $\abs{S_{v}(\ell')}\le 2C_0\ell$,
  the runtime in total is at most $O(2^{C_1\ell}+n\ell r^{-\ell'}\log n)$ with high probability.
  We choose $\ell=\frac{1-c}{C_1}\log n$ for $c=\frac{\log r}{\log r+2C_1}\in[0,1]$,
  so that $\ell' \geq \ell /2 = \frac{1-c}{2C_1} \log n$  and the total runtime is $O(n^{1-c}+n^{1-(1-c)\log r/(2C_1)}\log^2 n)=\widetilde{O}(n^{1-c})$.
\end{proof}

Now we are ready to prove \Cref{lattice-counting}. 

\begin{proof}[Proof of Theorem~\ref{lattice-counting}]
  Using the same trick at the end of the proof of \Cref{small-lambda},
  it suffices to give an algorithm that runs in time $\widetilde{O}(n^{2-c}/\eps^2)$.

  We are going to use \eqref{q-spin-decomposition} to do a self-reduction.
  First we construct the target configuration $\sigma$ adaptively.
  Given $\sigma$ on $v_1,\ldots,v_{i-1}$,
  we want to choose $\sigma_{v_i}$ to be $k\in [q]$ with the largest marginal.
  In other words, $\sigma_{v_i}=\argmax_{k\in[q]}\mu_{v_i}^{\sigma_i}(k)$ for each $i$,
  where $\sigma_i$ is what has been constructed so far, namely $\sigma_{v_1},\ldots,\sigma_{v_{i-1}}$.
  Of course, this step cannot be done exactly.
  Instead, we may fix a constant $t = t(C,r,q)$ such that $Cr^{-t} \leq \frac{1}{2q}$, fix an arbitrary boundary  configuration $\tau$ on $S^{\sigma_i}_{v_i}(t)$ and then pick $k\in[q]$ that maximises $\mu^{\sigma_i,\tau}_{v_i}(k)$.
  SSM guarantees that
  $\mu_{v_i}^{\sigma_i}(\sigma_{v_i})\ge 1/2q$, where $\sigma_{v_i} = k$.
  This step takes constant time as $t$ is a constant.

  The rest of the proof is very similar to that of \Cref{small-lambda}.
  Set $N\defeq\lceil10 e^{4q^2}/\eps_0^2\rceil$ where $\eps_0=\eps/2$. 
  We compute $X=\prod_{i=1}^n\widetilde{Z}_i$ where each $\widetilde{Z}_i$ is from \Cref{lem:lattice-estimator} plugging in $v_i$ and $\sigma_i$. Due to the decomposition \eqref{q-spin-decomposition} we have
  \[\Ex[X]=\Ex\left[\prod_{i=1}^n\widetilde{Z}_i\right]=\prod_{i=1}^n\Ex\left[\widetilde{Z}_i\right]=\prod_{i=1}^n\mu_{v_i}^{\sigma_i}(\sigma_{v_i})=\mu(\sigma).\]
  We also compute $w(\sigma)$ which can be done in $O(n)$ on a planar graph with quadratic growth.
  By \Cref{lem:lattice-estimator}, the time to generate one $X$ is $O(n^{2-c}\polylog{n})$ with high probability. 
  We bound $\Var{}{X}$ as follows
  \begin{align*}
    \frac{\Var{}{X}}{(\Ex[X])^2}&=\frac{\Ex[X^2]}{(\Ex[X])^2}-1=\frac{\prod_{i=1}^n\Ex[\widetilde{Z}_i^2]}{\prod_{i=1}^n\Ex[\widetilde{Z}_i]^2}-1
    =\prod_{i=1}^n\left(1+\frac{\Var{}{\widetilde{Z}_i}}{(\Ex[\widetilde{Z}_i])^2}\right)-1\\
    &\leq\left(1+\frac{4q^2}{n}\right)^n-1\leq e^{4q^2},
  \end{align*}
  where we use $\mu_{v_i}^{\sigma_i}(\sigma_{v_i})\ge 1/2q$ for any $i\in[n]$.
  Let $\widetilde{X}$ be the average of $N$ samples of $X$. Then $\Var{}{\widetilde{X}}=\frac{\Var{}{X}}{N}\leq\frac{e^{4q^2}}{N\cdot Z(G)^2}$. 
  By Chebyshev's inequality,
  \begin{align*}
    \Pr\left[\abs{\widetilde{X}-\frac{1}{Z(G)}}\ge\frac{\epsnot}{Z(G)}\right]\le\frac{\Var{}{\widetilde{X}}}{\frac{\epsnot^2}{Z(G)^2}} 
    \le \frac{e^{4q^2}}{N\cdot Z(G)^2} \cdot \frac{Z(G)^2}{\epsnot^2}\le \frac{1}{10}.
  \end{align*}
  Thus, with probability at least $9/10$,
  we have that $\frac{1-\epsnot}{Z(G)}\le \widetilde{X} \le \frac{1+\epsnot}{Z(G)}$.
  Finally, we output $\widetilde{Z}=w(\sigma)/\widetilde{X}$.
  Since \Cref{def:quadratic-growth} implies a constant degree bound, the graph is sparse and $w(\sigma)$ can be computed in $O(n)$ time.
  To make sure that the algorithm runs within the time bound $O\left(\frac{n^{2-c}}{\eps^2}\right)$, 
  we truncate the algorithm if it runs overtime and output an arbitrary value in that case.
  This truncated version can be coupled with the untruncated algorithm with probability at least $7/8$,
  and its output $\widetilde{Z}$ satisfies $1-\eps\le \frac{\widetilde{Z}}{Z(G)}\le 1+\eps$ with probability at least $3/4$.
\end{proof}

%

\subsection{Bounded-radius circle packing}\label{sec:circle-packing}

Here we show that \Cref{lattice-counting} applies to any planar graph with bounded-radius circle packings.
We begin with the definition of a circle packing. 
\begin{definition}
A \emph{circle packing} is a collection $\mathcal{C}$ of interior-disjoint circles over the 2-dim-ensional plane. 
A \emph{tangency graph} of a circle packing is a graph having a vertex for each circle, and an edge between two vertices if and only if the two corresponding circles are tangent. 
\end{definition}

The Koebe-Andreev-Thurston \emph{circle packing theorem} states the following.
\begin{theorem}
For every connected locally finite simple planar graph $\mathbb{G}$, there exists a circle packing whose tangency graph is (isomorphic to) $\mathbb{G}$. 
\end{theorem}

We are concerned with the radius of the circles used in the packing, especially the ratio between the smallest and largest ones. 
\begin{definition}
  A locally finite simple planar graph $\mathbb{G}$ is said to have an \emph{$R$-bounded-radius circle packing} ($R$-BRCP) for some constant $R>0$, if there exists a circle packing $\mathcal{C}$ whose tangency graph is (isomorphic to) $\mathbb{G}$ such that
\[
\frac{\inf\limits_{\odot\in\mathcal{C}}r_{\odot}}{\sup\limits_{\odot\in\mathcal{C}}r_{\odot}}\geq R
\]
where $r_{\odot}$ denotes the radius of a circle $\odot$ in the packing. 
\end{definition}

\begin{figure}
	\definecolor{myorange}{RGB}{252,141,98}
	\definecolor{mygreen}{RGB}{102,194,165}
	\definecolor{myblue}{RGB}{141,160,203}
\centering
\begin{minipage}{0.32\textwidth}
\centering
\resizebox{\linewidth}{!}{
\begin{tikzpicture}[rotate=90]
  \clip (-0.5,{-sqrt(3)/6}) rectangle (12.5,{18.5/3*sqrt(3)});

\foreach \j in {0,...,6} {
  \def\y{{\j*2}}
  \draw [thick] (-0.5,\y) -- (12.5,\y);
}
\foreach \i in {0,...,6} {
  \def\x{{\i*2}}
  \draw [thick] (\x,{-sqrt(3)/3*0.5}) -- (\x,{18.5/3*sqrt(3)});
}

\foreach \i in {0,...,6} {
  \foreach \j in {0,...,6} {
    \def\x{{2*\i}}
    \def\y{{2*\j}}
    \filldraw[color=myorange, fill=myorange, fill opacity=0.15, ultra thick](\x,\y) circle (1);
    \node[circle,fill=myorange,draw=black,inner sep=2pt] () at (\x,\y) {};
  }
}
\end{tikzpicture}
}
\end{minipage}
\begin{minipage}{0.32\textwidth}
\centering
\resizebox{\linewidth}{!}{
\begin{tikzpicture}[rotate=90]
  \clip (-0.5,{-sqrt(3)/6}) rectangle (12.5,{18.5/3*sqrt(3)});

\foreach \j in {0,...,6} {
  \def\y{{\j*sqrt(3)}}
  \draw [thick] (-0.5,\y) -- (12.5,\y);
}
\foreach \i in {0,...,4} {
  \def\x{{\i*3}}
  \draw [thick] (\x,{-sqrt(3)/3*0.5}) -- (\x,{18.5/3*sqrt(3)});
}

\foreach \i in {0,...,5} {
  \def\x{{(6*\i-12-0.5)*sqrt(3)/3}}
  \def\y{{(6*\i+0.5)*sqrt(3)/3}}
  \draw [thick] (-0.5,\x) -- (12.5,\y);
  \draw [thick] (12.5,\x) -- (-0.5,\y);
}

\foreach \i in {0,...,9} {
  \def\x{{2*\i-37/6}}
  \def\y{{2*\i+1/6}}
  \draw [thick] (\x,{37/6*sqrt(3)}) -- (\y,{-1/6*sqrt(3)});
  \draw [thick] (\y,{37/6*sqrt(3)}) -- (\x,{-1/6*sqrt(3)});
}

\foreach \i in {0,...,4} {
  \foreach \j in {0,...,6} {
    \def\x{{3*\i}}
    \def\y{{sqrt(3)*\j}}
    \pgfmathparse{{int(mod(\i+\j,2))}}
    \ifnum \pgfmathresult=0
      {
        \filldraw[color=myorange, fill=myorange, fill opacity=0.15, ultra thick](\x,\y) circle ({(1+sqrt(3))/2});
        \node[circle,fill=myorange,draw=black,inner sep=2pt] () at (\x,\y) {};
      }
    \fi
  }
}
\foreach \i in {-1,...,4} {
  \foreach \j in {0,...,6} {
    \def\xa{{3*\i+1}}
    \def\xb{{3*\i+2}}
    \def\y{{sqrt(3)*\j}}
    \pgfmathparse{{int(mod(\i+\j,2))}}
    \ifnum \pgfmathresult=0
      {\filldraw[color=mygreen, fill=mygreen, fill opacity=0.2, ultra thick](\xb,\y) circle ({(3-sqrt(3))/2});
      \node[circle,fill=mygreen,draw=black,inner sep=2pt] () at (\xb,\y) {};}
    \else
      {\filldraw[color=mygreen, fill=mygreen, fill opacity=0.2, ultra thick](\xa,\y) circle ({(3-sqrt(3))/2});
      \node[circle,fill=mygreen,draw=black,inner sep=2pt] () at (\xa,\y) {};}
    \fi
  }
}
\foreach \i in {0,...,4} {
  \foreach \j in {0,...,6} {
    \def\x{{3*\i}}
    \def\y{{sqrt(3)*\j}}
    \pgfmathparse{{int(mod(\i+\j,2))}}
    \ifnum \pgfmathresult=1
      {\filldraw[color=myblue, fill=myblue, fill opacity=0.2, ultra thick](\x,\y) circle ({(sqrt(3)-1)/2});
      \node[circle,fill=myblue,draw=black,inner sep=2pt] () at (\x,\y) {};}
    \fi
  }
}
\foreach \i in {0,...,3} {
  \foreach \j in {0,...,5} {
    \def\x{{3*\i+1.5}}
    \def\y{{sqrt(3)*\j+sqrt(3)/2}}
    \filldraw[color=myblue, fill=myblue, fill opacity=0.2, ultra thick](\x,\y) circle ({(sqrt(3)-1)/2});
    \node[circle,fill=myblue,draw=black,inner sep=2pt] () at (\x,\y) {};
  }
}
\end{tikzpicture}
}
\end{minipage}
\begin{minipage}{0.32\textwidth}
\centering
\resizebox{\linewidth}{!}{
\begin{tikzpicture}

\filldraw[color=myorange, fill=myorange, fill opacity=0.15, ultra thick](0:0) circle (1);
\node[circle,fill=myorange,draw=black,inner sep=2pt] () at (0:0) {};

\foreach \i in {0,...,2} {
    \def\t{360/3*\i+360/6}
    \def\r{2}
    \def\s{1}
    \filldraw[color=myblue, fill=myblue, fill opacity=0.2, ultra thick](\t:\r) circle (\s);
    \node[circle,fill=myblue,draw=black,inner sep=2pt] () at (\t:\r) {};
    \begin{scope}[on background layer]
      \draw[thick] (\t:\r) -- (0,0);
    \end{scope}
}

\foreach \i in {0,...,5} {
    \def\t{360/6*\i+360/12}
    \def\r{{sqrt(3)*2}}
    \def\s{1}
    \filldraw[color=mygreen, fill=mygreen, =0.5, fill opacity=0.2, ultra thick](\t:\r) circle (\s);
    \node[circle,fill=mygreen,draw=black,inner sep=2pt] () at (\t:\r) {};
    \begin{scope}[on background layer]
      \def\tt{{360/3*int(\i/2)+360/6}}
      \def\rr{{2}}
      \draw[thick] (\t:\r) -- (\tt:\rr);
    \end{scope}
}

\foreach \i in {0,...,11} {
    \def\t{360/12*\i+360/24}
    \def\r{5.13384}
    \def\s{1}
    \filldraw[color=myorange, fill=myorange, fill opacity=0.15, ultra thick](\t:\r) circle (\s);
    \node[circle,fill=myorange,draw=black,inner sep=2pt] () at (\t:\r) {};
    \begin{scope}[on background layer]
      \def\tt{{360/6*int(\i/2)+360/12}}
      \def\rr{{sqrt(3)*2}}
      \draw[thick] (\t:\r) -- (\tt:\rr);
    \end{scope}
}

\foreach \i in {0,...,23} {
    \def\t{360/24*\i+360/48}
    \def\r{6.63}
    \def\s{0.6667}
    \filldraw[color=myblue, fill=myblue, fill opacity=0.2, ultra thick](\t:\r) circle (\s);
    \node[circle,fill=myblue,draw=black,inner sep=2pt] () at (\t:\r) {};
    \begin{scope}[on background layer]
      \def\tt{{360/12*int(\i/2)+360/24}}
      \def\rr{5.13384}
      \draw[thick] (\t:\r) -- (\tt:\rr);
    \end{scope}
}

\foreach \i in {0,...,47} {
    \def\t{360/48*\i+360/96}
    \def\r{7.6}
    \def\s{0.4}
    \filldraw[color=mygreen, fill=mygreen, fill opacity=0.2, ultra thick](\t:\r) circle (\s);
    \node[circle,fill=mygreen,draw=black,inner sep=2pt] () at (\t:\r) {};
    \begin{scope}[on background layer]
      \def\tt{{360/24*int(\i/2)+360/48}}
      \def\rr{6.63}
      \draw[thick] (\t:\r) -- (\tt:\rr);
    \end{scope}
}

\foreach \i in {0,...,95} {
    \def\t{360/96*\i+360/192}
    \def\r{8.2}
    \def\s{0.24}
    \filldraw[color=myorange, fill=myorange, fill opacity=0.15, ultra thick](\t:\r) circle (\s);
    \node[circle,fill=myorange,draw=black,inner sep=2pt] () at (\t:\r) {};
    \begin{scope}[on background layer]
      \def\tt{{360/48*int(\i/2)+360/96}}
      \def\rr{7.6}
      \draw[thick] (\t:\r) -- (\tt:\rr);
    \end{scope}
}

\end{tikzpicture}
}
\end{minipage}
\par
\begin{minipage}{0.32\textwidth}
\centering
\bigskip(a)
\end{minipage}
\begin{minipage}{0.32\textwidth}
\centering
\bigskip(b)
\end{minipage}
\begin{minipage}{0.32\textwidth}
\centering
\bigskip(c)
\end{minipage}

\caption{Circle packings of some lattices. (a): $\mathbb{Z}^2$ grid, $R=1$. (b): Kisrhombille tiling, $R=2-\sqrt{3}$. (c): degree-$3$ Bethe lattice, $R=0$.}
\label{fig:packing}
\end{figure}
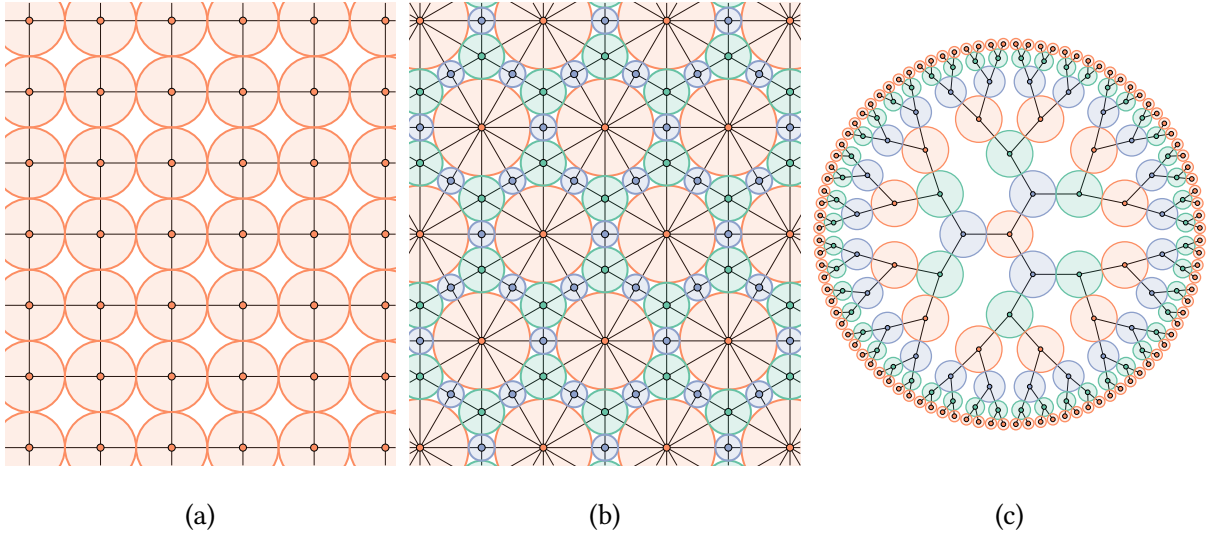

Three examples are given in \Cref{fig:packing}. 
The $\mathbb{Z}^2$ grid can be naturally packed by unit disks, leading to $R=1$. 
Such a graph is called a ``penny graph''. 
The $3,6$-kisrhombille tiling is a tiling of the 2-dimensional plane by $\pi/6$-$\pi/3$-$\pi/2$ triangles. 
This lattice can be packed by circles of radii $1,2\sqrt{3}-3,2-\sqrt{3}$, so $R=2-\sqrt{3}$. 
The degree-$3$ Bethe lattice, also known as the infinite $3$-regular tree, can be drawn as a planar graph on the 2-dimensional plane. 
However, the neighbourhood growth is so fast that 
$R=0$. 

Fix the underlying graph $\mathbb{G}$ and its $R$-BRCP $\mathcal{C}$. 
Without loss of generality, we assume the diameter of the largest circle in $\+C$ is $1$.
Thus, the radius of an arbitrary circle in $\+C$ is between $R/2$ and $1/2$.
Let $G$ be a finite subgraph of $\=G$.
Here we need to distinguish the graph distance in $G$ and the geometric distance (the Euclidean distance $\left\|\cdot\right\|_2$ between the center of their corresponding disks on the 2-dimensional plane).
For two vertices $u$ and $v$,
we use $\dist_G(u,v)$ to denote their graph distance,
and use $\norm{2}{u-v}$ to denote their geometric distance.
Note that $\dist_{\=G}(u,v)\ge\norm{2}{u-v}$ and $\dist_{G}(u,v)\ge\dist_{\=G}(u,v)$.

For any vertex $v$ and $u$ in the $\ell$-ball $B_v(\ell)$ in $G$,
$\norm{2}{u-v}\le\dist_G(u,v)\le \ell$.
The disk $\odot_u$ corresponding to $u$ must be contained completely in the circle centered at $u$ with radius $\ell+1/2$.
By considering the area they cover, 
\begin{align*}
  \abs{B_v(\ell)}\le \frac{\pi(\ell+1/2)^2}{\pi (R/2)^2} = O(\ell^2/R^2).
\end{align*}
Thus, any family of subgraphs of $\=G$ has quadratic growth,
where the growth constant depends on $R$.
Together with \Cref{lattice-counting}, we have the following corollary.

\begin{corollary}\label{cor:BRCP}
  Let $\mathbb{G}$ be a locally finite simple planar graph, together with an $R$-BRCP where $R>0$ is a constant. 
Let $\+{G}$ be a family of subgraphs of $\mathbb{G}$,
and ${\mathbf{A}}, {\mathbf{b}}$ specify a $q$-spin system that exhibits SSM with exponential decay on $\mathcal{G}$. 
Then there exists an FPRAS that takes a graph $G\in\mathcal{G}$ as an input and estimates the partition function of the spin system on $G$ in time $\widetilde{O}\big(\big(\frac{n}{\eps}\big)^{2-c}\big)$. 
Here, $n=|V(G)|$, and $c>0$ is a constant depending on $q$, decay rate of SSM, and $R$. 
\end{corollary}

\begin{remark}
  The algorithm does not need to know the circle packing,
  as long as an $R$-BRCP exists.

  On a separate note, although a good approximation of the circle packing of a finite planar graph can be found in near linear time \cite{DLQ20}, 
  its output does not optimise the radius ratio.
  It is not clear how to generate a circle packing with a constant approximation of the optimal radius ratio. 
  In the extreme, it is $\NP$-hard to decide if a given graph $G$ (without geometric positions) is a penny graph, namely admitting a circle packing using unit circles \cite{EW96}, even if $G$ is restricted to be a tree \cite{BDLRST15}. 
\end{remark}

\subsection{Polynomial-growth graphs} \label{sec:poly-growth}

Our method goes beyond planar graphs with quadratic growth rate.
For any graph with a polynomial growth rate, we have a speed-up that is faster than any polylog factors.

\begin{definition}  \label{def:poly-growth}
  A graph family $\=G$ has \emph{polynomial growth},
  if there are constants $C$ and $d$ such that for any $G=(V,E)\in\=G$, $v\in V$, and any integer $\ell>0$,
  $\abs{B_v(\ell)}\le C_0\ell^d$.
\end{definition}

Examples of graphs with polynomial growth include finite subgraphs of the $d$-dimensional integer lattice $\=Z^d$.
Again, by taking $\ell=1$, \Cref{def:poly-growth} implies that the maximum degree is no larger than~$C_0$.

\begin{theorem} \label{thm:polylog}
  Let $\=G$ be a family of graphs with polynomial growth (assume the rate is $C_0\ell^d$).
  Let ${\mathbf{A}}$ and ${\mathbf{b}}$ specify a $q$-state spin system, which exhibits SSM with decay rate $Cr^{-\ell}$ on $\=G$.
  Then there is a constant $c>0$ such that there exists an FPRAS for the partition function of $G\in\=G$ with $n$ vertices with run-time $\widetilde{O}\left(\frac{n^2}{\eps^2 2^{c(\log \frac{n}{\eps})^{1/d}}}\right)$.
  The constant $c$ depends on $C_0$, $q$, and~$r$.
\end{theorem}

\Cref{thm:polylog} is the detailed version of \Cref{thm:main3}.

In comparison to \Cref{lattice-counting},
the proof of \Cref{thm:polylog} needs only a few small tweaks.
Let $\ell$ be a parameter we will choose later,
and our estimator is still set by using a random boundary condition on $S_v(\ell)$ to estimate the marginal at $v$.
Note that we no longer need to find $\ell'$ for a smaller boundary.
The main difference is in \Cref{YZ-table},
where we no longer have linear local tree-width.
Instead, we have to create the table by brute-force enumeration.
There are $q^{C_0\ell^d}$ possible boundary conditions,
and the overall time cost for creating the table is $O(q^{2C_0\ell^d})$.

We use the same estimator as in \Cref{lem:lattice-estimator}.
To reduce the variance of our estimator to $1/n$, we need $nC^2r^{-2\ell}$ samples,
each of which can be looked up quickly using the table.
Let $\ell=\frac{0.99(\log n)^{1/d}}{2C_0\log q}$.
The overall time cost is
\begin{align*}
  \widetilde{O}\left( \frac{n}{\eps^2}\left( q^{2C\ell^d}+ nr^{-2\ell}\right) \right)
  =\widetilde{O}\left( \frac{n}{\eps^2}\left(n^{0.99}+\frac{n}{2^{c(\log n)^{1/d}}}\right) \right)
  =\widetilde{O}\left(\frac{n^2}{\eps^2 2^{c(\log n)^{1/d}}}\right),
\end{align*}
where $c=\frac{0.99\log r}{C_0\log q}$.
To finish the proof of \Cref{thm:polylog},
we employ the trick at the end of the proof of \Cref{small-lambda} once again.
Note that the factor $2^{c(\log \frac{n}{\eps})^{1/d}}$ grows faster than $\big(\log \frac{n}{\eps}\big)^k$ for any constant $k>0$.

\section*{Acknowledgement}
We would like to thank Chunyang Wang for pointing out how to shave a factor of $e$ from \Cref{lem:gwb-terminate}. 

\printbibliography

@inproceedings{BRY20,
  author    = {Amartya Shankha Biswas and
               Ronitt Rubinfeld and
               Anak Yodpinyanee},
  title     = {Local Access to Huge Random Objects Through Partial Sampling},
  booktitle = {11th Innovations in Theoretical Computer Science Conference, {ITCS}},
  series    = {LIPIcs},
  volume    = {151},
  pages     = {27:1--27:65},
  publisher = {Schloss Dagstuhl - Leibniz-Zentrum f{\"{u}}r Informatik},
  year      = {2020},
  doi       = {10.4230/LIPICS.ITCS.2020.27}
}

@inproceedings{AnandFFGW24,
  author       = {Konrad Anand and
                  Weiming Feng and
                  Graham Freifeld and
                  Heng Guo and
                  Jiaheng Wang},
  editor       = {Karl Bringmann and
                  Martin Grohe and
                  Gabriele Puppis and
                  Ola Svensson},
  title        = {Approximate Counting for Spin Systems in Sub-Quadratic Time},
  booktitle    = {51st International Colloquium on Automata, Languages, and Programming},
  series       = {LIPIcs},
  volume       = {297},
  pages        = {11:1--11:20},
  publisher    = {Schloss Dagstuhl - Leibniz-Zentrum f{\"{u}}r Informatik},
  year         = {2024},
  url          = {https://doi.org/10.4230/LIPIcs.ICALP.2024.11},
  doi          = {10.4230/LIPICS.ICALP.2024.11},
  bibsource    = {dblp computer science bibliography, https://dblp.org}
}

@article{PR17,
  author       = {Viresh Patel and
                  Guus Regts},
  title        = {Deterministic Polynomial-Time Approximation Algorithms for Partition
                  Functions and Graph Polynomials},
  journal      = {{SIAM} J. Comput.},
  volume       = {46},
  number       = {6},
  pages        = {1893--1919},
  year         = {2017},
  doi          = {10.1137/16M1101003}
}

@book{Bar16,
  author       = {Alexander I. Barvinok},
  title        = {Combinatorics and Complexity of Partition Functions},
  series       = {Algorithms and combinatorics},
  volume       = {30},
  publisher    = {Springer},
  year         = {2016},
  doi          = {10.1007/978-3-319-51829-9}
}

@inproceedings{CHLP23,
  author       = {Ruoxu Cen and
                  William He and
                  Jason Li and
                  Debmalya Panigrahi},
  editor       = {David P. Woodruff},
  title        = {Beyond the Quadratic Time Barrier for Network Unreliability},
  booktitle    = {Proceedings of the 2024 {ACM-SIAM} Symposium on Discrete Algorithms,
                  {SODA} 2024},
  pages        = {1542--1567},
  publisher    = {{SIAM}},
  year         = {2024},
  url          = {https://doi.org/10.1137/1.9781611977912.62},
  doi          = {10.1137/1.9781611977912.62},
  bibsource    = {dblp computer science bibliography, https://dblp.org}
}

@inproceedings{Wei06,
	author = {Dror Weitz},
	booktitle = {Proceedings of the 38th Annual {ACM} Symposium on Theory of Computing, {STOC}},
	pages = {140--149},
	publisher = {{ACM}},
	title = {Counting independent sets up to the tree threshold},
	year = {2006},
    doi = {10.1145/1132516.1132538}
	}

@inproceedings{LLY13,
	author = {Liang Li and Pinyan Lu and Yitong Yin},
	booktitle = {Proceedings of the Twenty-Fourth Annual {ACM-SIAM} Symposium on Discrete
                  Algorithms, {SODA}},
	note = {Full version from arXiv at \texttt{abs/1111.7064}.},
	pages = {67--84},
	publisher = {{SIAM}},
	title = {Correlation Decay up to Uniqueness in Spin Systems},
	year = {2013},
    doi = {10.1137/1.9781611973105.5}
	}

@article{GSV16,
	author = {Galanis, Andreas and {\v{S}}tefankovi{\v{c}}, Daniel and Vigoda, Eric},
	journal = {Combin. Probab. Comput.},
	number = {4},
	pages = {500--559},
	title = {Inapproximability of the partition function for the antiferromagnetic {I}sing and hard-core models},
	volume = {25},
	year = {2016},
    doi = {10.1017/S0963548315000401}
	}

@article{GSV15,
  author       = {Andreas Galanis and
                  Daniel {\v{S}}tefankovi{\v{c}} and
                  Eric Vigoda},
  title        = {Inapproximability for Antiferromagnetic Spin Systems in the Tree Nonuniqueness
                  Region},
  journal      = {J. {ACM}},
  volume       = {62},
  number       = {6},
  pages        = {50:1--50:60},
  year         = {2015},
  doi          = {10.1145/2785964}
}

@inproceedings{CLV21,
  author       = {Zongchen Chen and
                  Kuikui Liu and
                  Eric Vigoda},
  title        = {Optimal mixing of {G}lauber dynamics: {E}ntropy factorization via high-dimensional
                  expansion},
  booktitle    = {53rd Annual {ACM} {SIGACT} Symposium on Theory of Computing, {STOC}},
  pages        = {1537--1550},
  publisher    = {{ACM}},
  year         = {2021},
  doi          = {10.1145/3406325.3451035}
}

@inproceedings{AJKPV22,
  author       = {Nima Anari and
                  Vishesh Jain and
                  Frederic Koehler and
                  Huy Tuan Pham and
                  Thuy{-}Duong Vuong},
  title        = {Entropic independence: optimal mixing of down-up random walks},
  booktitle    = {54th Annual {ACM} {SIGACT} Symposium on Theory of Computing, {STOC}},
  pages        = {1418--1430},
  publisher    = {{ACM}},
  year         = {2022},
  doi          = {10.1145/3519935.3520048}
}

@inproceedings{CFYZ22,
  author       = {Xiaoyu Chen and
                  Weiming Feng and
                  Yitong Yin and
                  Xinyuan Zhang},
  title        = {Optimal mixing for two-state anti-ferromagnetic spin systems},
  booktitle    = {63rd {IEEE} Annual Symposium on Foundations of Computer Science, {FOCS}},
  pages        = {588--599},
  publisher    = {{IEEE}},
  year         = {2022},
  doi          = {10.1109/FOCS54457.2022.00062}
}

@article{FGY22,
  author       = {Weiming Feng and
                  Heng Guo and
                  Yitong Yin},
  title        = {Perfect sampling from spatial mixing},
  journal      = {Random Struct. Algorithms},
  volume       = {61},
  number       = {4},
  pages        = {678--709},
  year         = {2022},
  doi          = {10.1002/RSA.21079}
}

@inproceedings{CE22,
  author       = {Yuansi Chen and
                  Ronen Eldan},
  title        = {Localization Schemes: {A} Framework for Proving Mixing Bounds for
                  Markov Chains (extended abstract)},
  booktitle    = {63rd {IEEE} Annual Symposium on Foundations of Computer Science, {FOCS}},
  pages        = {110--122},
  publisher    = {{IEEE}},
  year         = {2022},
  doi          = {10.1109/FOCS54457.2022.00018}
}

@article{HS07,
    AUTHOR = {Hayes, Thomas P. and Sinclair, Alistair},
     TITLE = {A general lower bound for mixing of single-site dynamics on
              graphs},
   JOURNAL = {Ann. Appl. Probab.},
    VOLUME = {17},
      YEAR = {2007},
    NUMBER = {3},
     PAGES = {931--952},
	   DOI = {10.1214/105051607000000104}
}

@article{CGPSSW18,
  author       = {Timothy Chu and
                  Yu Gao and
                  Richard Peng and
                  Sushant Sachdeva and
                  Saurabh Sawlani and
                  Junxing Wang},
  title        = {Graph Sparsification, Spectral Sketches, and Faster Resistance Computation
                  via Short Cycle Decompositions},
  journal      = {{SIAM} J. Comput.},
  volume       = {52},
  number       = {6},
  pages        = {S18--85},
  year         = {2023},
  url          = {https://doi.org/10.1137/19m1247632},
  doi          = {10.1137/19M1247632},
  bibsource    = {dblp computer science bibliography, https://dblp.org}
}

@article{SST14,
	author = {Sinclair, Alistair and Srivastava, Piyush and Thurley, Marc},
	journal = {J. Stat. Phys.},
	number = {4},
	pages = {666--686},
	title = {Approximation algorithms for two-state anti-ferromagnetic spin systems on bounded degree graphs},
	volume = {155},
	year = {2014}}

@article{SS14,
	author = {Sly, Allan and Sun, Nike},
	journal = {Ann. Probab.},
	number = {6},
	pages = {2383--2416},
	title = {Counting in two-spin models on {$d$}-regular graphs},
	volume = {42},
	year = {2014},
	doi = {10.1214/13-AOP888}}

@article{CC17,
	author = {Jin{-}Yi Cai and Xi Chen},
	journal = {J. {ACM}},
	number = {3},
	pages = {19:1--19:39},
	title = {Complexity of Counting {CSP} with Complex Weights},
	volume = {64},
	year = {2017},
    doi = {10.1145/2822891}
}

@article{JVV86,
  author    = {Mark Jerrum and
               Leslie G. Valiant and
               Vijay V. Vazirani},
  title     = {Random Generation of Combinatorial Structures from a Uniform Distribution},
  journal   = {Theor. Comput. Sci.},
  volume    = {43},
  pages     = {169--188},
  year      = {1986},
  doi       = {10.1016/0304-3975(86)90174-X}
}

@article{Val79,
  author       = {Leslie G. Valiant},
  title        = {The Complexity of Computing the Permanent},
  journal      = {Theor. Comput. Sci.},
  volume       = {8},
  pages        = {189--201},
  year         = {1979},
  doi          = {10.1016/0304-3975(79)90044-6}
}

@article{JS89,
  author       = {Mark Jerrum and
                  Alistair Sinclair},
  title        = {Approximating the Permanent},
  journal      = {{SIAM} J. Comput.},
  volume       = {18},
  number       = {6},
  pages        = {1149--1178},
  year         = {1989},
  doi          = {10.1137/0218077}
}

@article{JS93,
  author       = {Mark Jerrum and
                  Alistair Sinclair},
  title        = {Polynomial-Time Approximation Algorithms for the {I}sing Model},
  journal      = {{SIAM} J. Comput.},
  volume       = {22},
  number       = {5},
  pages        = {1087--1116},
  year         = {1993},
  doi          = {10.1137/0222066}
}

@article{Hub15,
    AUTHOR = {Huber, Mark},
     TITLE = {Approximation algorithms for the normalizing constant of
              {G}ibbs distributions},
   JOURNAL = {Ann. Appl. Probab.},
    VOLUME = {25},
      YEAR = {2015},
    NUMBER = {2},
     PAGES = {974--985},
	   DOI = {10.1214/14-AAP1015}
}

@article{SVV09,
  author    = {Daniel {\v{S}}tefankovi{\v{c}} and
               Santosh S. Vempala and
               Eric Vigoda},
  title     = {Adaptive simulated annealing: {A} near-optimal connection between
               sampling and counting},
  journal   = {J. {ACM}},
  volume    = {56},
  number    = {3},
  pages     = {18:1--18:36},
  year      = {2009},
  doi       = {10.1145/1516512.1516520}
}

@inproceedings{Kol18,
  author    = {Vladimir Kolmogorov},
  title     = {A Faster Approximation Algorithm for the {G}ibbs Partition Function},
  booktitle = {Conference On Learning Theory, {COLT}},
  series    = {Proceedings of Machine Learning Research},
  volume    = {75},
  pages     = {228--249},
  publisher = {{PMLR}},
  year      = {2018}
}

@inproceedings{YZ13,
  author       = {Yitong Yin and
                  Chihao Zhang},
  title        = {Approximate Counting via Correlation Decay on Planar Graphs},
  booktitle    = {Proceedings of the Twenty-Fourth Annual {ACM-SIAM} Symposium on Discrete
                  Algorithms, {SODA}},
  pages        = {47--66},
  publisher    = {{SIAM}},
  year         = {2013},
  doi          = {10.1137/1.9781611973105.4}
}

@article{AJ22,
  author       = {Konrad Anand and
                  Mark Jerrum},
  title        = {Perfect Sampling in Infinite Spin Systems Via Strong Spatial Mixing},
  journal      = {{SIAM} J. Comput.},
  volume       = {51},
  number       = {4},
  pages        = {1280--1295},
  year         = {2022},
  doi          = {10.1137/21M1437433}
}

@phdthesis{Guo15,
author={Heng Guo},
title={Complexity Classification of Exact and Approximate Counting Problems},
school  = "University of Wisconsin - Madison",
year={2015},
}

@inproceedings{DH04,
  author       = {Erik D. Demaine and
                  Mohammad Taghi Hajiaghayi},
  title        = {Equivalence of local treewidth and linear local treewidth and its
                  algorithmic applications},
  booktitle    = {Proceedings of the Fifteenth Annual {ACM-SIAM} Symposium on Discrete
                  Algorithms, {SODA}},
  pages        = {840--849},
  publisher    = {{SIAM}},
  year         = {2004}
}

@article{Epp00,
  author       = {David Eppstein},
  title        = {Diameter and Treewidth in Minor-Closed Graph Families},
  journal      = {Algorithmica},
  volume       = {27},
  number       = {3},
  pages        = {275--291},
  year         = {2000},
  doi          = {10.1007/S004530010020}
}

@inproceedings{FGWWY22,
  author       = {Weiming Feng and
                  Heng Guo and
                  Chunyang Wang and
                  Jiaheng Wang and
                  Yitong Yin},
  title        = {Towards derandomising {M}arkov chain {M}onte {C}arlo},
  booktitle    = {64th {IEEE} Annual Symposium on Foundations of Computer Science, {FOCS}},
  pages        = {1963--1990},
  publisher    = {{IEEE}},
  year         = {2023},
  doi          = {10.1109/FOCS57990.2023.00120}
}

@inproceedings{DLQ20,
  author       = {Sally Dong and
                  Yin Tat Lee and
                  Kent Quanrud},
  title        = {Computing Circle Packing Representations of Planar Graphs},
  booktitle    = {Proceedings of the 2020 {ACM-SIAM} Symposium on Discrete Algorithms, {SODA}},
  pages        = {2860--2875},
  publisher    = {{SIAM}},
  year         = {2020},
  doi          = {10.1137/1.9781611975994.174}
}

@article{EW96,
  author       = {Peter Eades and
                  Sue Whitesides},
  title        = {The Logic Engine and the Realization Problem for Nearest Neighbor
                  Graphs},
  journal      = {Theor. Comput. Sci.},
  volume       = {169},
  number       = {1},
  pages        = {23--37},
  year         = {1996},
  doi          = {10.1016/S0304-3975(97)84223-5}
}

@inproceedings{BDLRST15,
  author       = {Clinton Bowen and
                  Stephane Durocher and
                  Maarten L{\"{o}}ffler and
                  Anika Rounds and
                  Andr{\'{e}} Schulz and
                  Csaba D. T{\'{o}}th},
  title        = {Realization of Simply Connected Polygonal Linkages and Recognition
                  of Unit Disk Contact Trees},
  booktitle    = {{GD}},
  series       = {Lecture Notes in Computer Science},
  volume       = {9411},
  pages        = {447--459},
  publisher    = {Springer},
  year         = {2015},
  doi          = {10.1007/978-3-319-27261-0\_37}
}

@article{KT16,
  author       = {Frank Kammer and
                  Torsten Tholey},
  title        = {Approximate tree decompositions of planar graphs in linear time},
  journal      = {Theor. Comput. Sci.},
  volume       = {645},
  pages        = {60--90},
  year         = {2016},
  doi          = {10.1016/J.TCS.2016.06.040}
}

\appendix

\section{Lazy marginal samplers}

\subsection{Specialised to hard-core models}
\label{sec:truncate-AJ}

\Cref{lem:truncate-AJ} is proved in this subsection. 
The single-site Anand-Jerrum algorithm adapts to the hard-core model as in \Cref{Alg:ssms-hardcore}. 
\begin{algorithm}  \LinesNotNumbered
  \caption{$\hcsampler(G,\lambda,(\Sigma,\sigma),v)$} \label{Alg:ssms-hardcore}
  \KwIn{a $\Delta$-degree graph $G$, fugacity $\lambda$, a set of vertices $\Sigma\subseteq V$ with a configuration $\sigma\in \Omega_{\Sigma}$, and vertex to sample $v\notin \Sigma$}
  \KwOut{the partial configuration passed in with a spin at
$v$: $(\Sigma, \sigma) \oplus (v, i)$ for some $i \in \{0,1\}$.}
Decrease the global timer $T\gets T-1$\;
\If{there exists $u\in\Sigma\cap N(v)$ such that $\sigma(u)=1$}{
  \Return $((\Sigma,\sigma)\oplus(v,0))$\;
}
Sample random $X\in \{\bot,0\}$ with $\Pr[X=\bot]=\lambda/(1+\lambda)$ and $\Pr[X=0]=1/(1+\lambda)$\;
\uIf{$X=\bot$}{
$(\Sigma',\sigma')\gets (\Sigma,\sigma)$\;
$Y\gets 1$\;
\ForAll{$u\in N(v)\backslash\Sigma$}{
  $(\Sigma',\sigma')\gets \hcsampler{}(G,\lambda,(\Sigma',\sigma'),u)$\; \label{Alg:line-branching}
  \lIf{$\sigma'(u)=1$}{$Y\gets 0$}
}
\Return $((\Sigma,\sigma)\oplus(v,Y))$\;
}
\Else
{
    \Return $((\Sigma,\sigma)\oplus(v,0))$\;
}
\end{algorithm} 

The correctness of the algorithm is summarised by the following theorem, adapted to our setting. 

\begin{theorem}[{\cite[Theorem 5.3]{AJ22}}]
Suppose $G$ is a graph with maximum degree bounded by $\Delta$, and $\lambda<\lambda_c(\Delta)$. 
If the untruncated algorithm $\hcsampler{}_{+\infty}(G,\lambda,(\Sigma,\sigma),v)$ terminates with probability $1$, then it generates a spin of $v$ according to the correct marginal distribution upon termination, provided that the partial configuration $(\Sigma,\sigma)$ is feasible. 
\end{theorem}

We remark that the correctness does not rely on the graph's neighbourhood growth being sub-exponential. 
However, the algorithm given here is a special case of that in \cite{AJ22}, where they look at an $\ell$-distance neighbourhood. 
Fixing $\ell=1$ as we do here results in the regime of fugacity~$\lambda$ being worse than the critical $\lambda_c$, as we will see very soon. 
The saving grace of \cite{AJ22} is that other $\ell$'s might be chosen in order to get to the critical regime, but this is at the cost of limiting the neighbourhood growth. 
Our main algorithm does not work up to the critical $\lambda_c$, so only the $1$-hop neighbourhood is considered. 



In \cite{AJ22}, the expected running time is studied and turns out to be a constant depending on the parameters of the model. 
However, we further need an exponential tail bound of the algorithm. 
This is done by the same idea of \cite[Section B.3]{FGWWY22}, 
though we do not truncate this algorithm as is done there. 
As soon as an exponential tail bound of running time is established, the algorithm then terminates with probability $1$ and hence is correct. 

We treat the algorithm as a branching process. 
Each time the algorithm recurses into its neighbourhood, it creates at most $\Delta-1$ new copies of the routine $\hcsampler{}$. 
Such branching happens with probability $p:=\lambda/(1+\lambda)$. 
This leads us to study the following Markovian process that stochastically dominates the actual branching process. 
Let $(X_t)_{t\in\mathbb{Z}_{\geq 0}}$ be a discrete Markov chain where $X_t\in\mathbb{Z}_{\geq 0}$ with initial state $X_0=1$. 
This chain has an absorbing barrier at $0$, and for any other $X_t>0$, the transition probability is given by
\begin{equation} \label{equ:gwb-def}
  X_{t+1}\leftarrow\begin{cases}
  X_t+\Delta-1 & \text{with probability } p; \\
  X_t-1 & \text{with probability } 1-p. 
  \end{cases}
\end{equation}

In the general case, the tail bound of this process is proved in \cite[Lemma B.12]{FGWWY22}, and this requires $\lambda\leq\frac{1}{2e\Delta-1}$ when specialised to the hard-core model. 
Here we provide a stronger analysis to remove the constant. 

\begin{lemma}\label{lem:gwb-terminate}
  Suppose $\lambda<\frac{1}{\Delta-1}$.
  For any $0<\varepsilon<1$, 
  let $T=\frac{2\Delta^2}{\left(\frac{\lambda}{1+\lambda}\Delta-1\right)^2}\log\frac{1}{\varepsilon}$. 
  Then with probability at most $\varepsilon$, the process $(X_t)$ defined by (\ref{equ:gwb-def}) does not terminate in $T$ rounds.  
\end{lemma}

\begin{proof}
Given $\{X_t\}_{t\in\mathbb{Z}_{\geq 0}}$, define an auxiliary process $\{Y_t\}_{t\in\mathbb{Z}_{\geq 0}}$ in the following way.
Let $Y_0=1$, and the transition probability is given by
\begin{equation}
  Y_{t+1}\leftarrow\begin{cases}
  Y_t+\frac{1}{1+\lambda}\Delta & \text{with probability } \frac{\lambda}{1+\lambda}; \\
  Y_t-\frac{\lambda}{1+\lambda}\Delta & \text{with probability } \frac{1}{1+\lambda}. 
  \end{cases}
\end{equation}
Then couple $X_t$ with $Y_t$ perfectly that, if $X_t$ increases then so does $Y_t$, and vice versa, till $X_t$ reaches the absorbing barrier. 
After this point, $Y_t$ just performs the above transition independently. 

Clearly, $\{Y_t\}$ is a martingale, and if $X_t>0$ is not absorbed then 
$Y_t=X_t+\left(\frac{\lambda}{1+\lambda}\Delta-1\right)t$. 
Also note that the regime on $\lambda$ ensures $\frac{\lambda}{1+\lambda}\Delta-1>0$. 
This allows us to bound the probability of $\{X_t\}$ not terminating after $T$ rounds by applying Azuma--Hoeffding inequality:  
\begin{align*}
\Pr[X_T>0]&=\Pr[X_T\geq X_0]=\Pr\left[Y_T-Y_0\geq T\cdot\left(\frac{\lambda}{1+\lambda}\Delta-1\right)\right]\leq\exp\left\{-\frac{T^2\left(\frac{\lambda}{1+\lambda}\Delta-1\right)^2}{2\Delta^2 T}\right\}=\varepsilon.
\end{align*}
\end{proof}

%
%

\Cref{lem:truncate-AJ} then follows by exactly the same argument as in \cite[Proof of Lemma B.10]{FGWWY22}, by noticing that the branching process $(X_t)$ stochastically dominates the number of `active' instances of $\hcsampler$, and using \Cref{lem:gwb-terminate}. 


\subsection{Generic sampler}

If we want to cover the whole strong spatial mixing regime but only work on amenable graphs, then we can invoke the original Anand-Jerrum algorithm, allowing us to do recursion at farther vertices rather than one-hop neighbours. 
For completeness, we include the algorithm here (\Cref{Alg:ssms} and \Cref{Alg:bdsplit}). 
Its running time tail bound is shown in \cite[Lemma B.10]{FGWWY22}. 

\begin{algorithm}  \LinesNotNumbered
  \caption{$\ssmsampler(\+S,(\Sigma,\sigma),v,r)$} \label{Alg:ssms}
  \KwIn{a spin system $\+S=(G,[q],{\mathbf{b}},{\mathbf{A}})$, a set of vertices $\Sigma\subseteq V$ with a configuration $\sigma\in \Omega_{\Sigma}$, a vertex to sample $v\notin \Sigma$, and a distance $r\in \mathbb{N}$}
  \KwOut{the partial configuration passed in with a spin at
$v$: $(\Sigma, \sigma) \oplus (v, i)$ for some $i \in [q]$.}
\For{$i\in [q]$}{
    $p_{v}^i\gets \min_{\tau\in \Omega_{S_{r}\setminus \Sigma}}\mu^{\sigma\oplus \tau}(i)$\;
}
$p_{v}^0\gets 1-\sum_{i\in [q]}p_{v}^i$\;
Sample a random value $X\in \{0,1,\ldots,q\}$ with $\Pr[X=i]=p_v^i$ for each $0\leq i\leq q$\;
\uIf{$X=0$}{
  $(\rho_1,\rho_2,\ldots,\rho_q)\gets \bdsplit(\+S,(\Sigma,\sigma),v,r,(p_v^{0},p_v^{1},p_{v}^2,\ldots,p_{v}^q))$\;
  Sample a random value $Y\in [q]$ with $\Pr[Y=i]=\rho_i$ for each $1\leq i\leq q$\;  
\Return $((\Sigma,\sigma)\oplus(v,Y))$\;
}
\Else
{
    \Return $((\Sigma,\sigma)\oplus(v,X))$\;
}
\end{algorithm} 

\begin{algorithm}  \LinesNotNumbered
  \caption{$\bdsplit(\+S,(\Sigma,\sigma),v,r,(p_v^0,p_v^1,p_v^2,\ldots,p_v^q))$} \label{Alg:bdsplit}
  \KwIn{a spin system $\+S=(G,[q],{\mathbf{b}},{\mathbf{A}})$, a set of vertices $\Sigma\subseteq V$ with a configuration $\sigma\in \Omega_{\Sigma}$, a vertex to sample $v\notin \Sigma$, a distance $r\in \mathbb{N}$, and a probability distribution $(p_v^0,p_v^1,p_v^2,\ldots,p_v^q)$}
  \KwOut{a probability distribution $(\rho_1,\rho_2,\ldots,\rho_q)$}
    Let $S_{r}(v)\gets\left\{u \mid \dist_G(u,v)=r\right\}$
    Give $S_{r}(v)\setminus \Sigma$ an arbitrary ordering $S_{r}(v)\setminus \Sigma=\{w_1,w_2,\ldots,w_m\}$\;
    $(\Sigma',\sigma')\gets (\Sigma,\sigma)$\;
    \For{$1\leq j\leq m$}{
        $(\Sigma',\sigma')\gets \ssmsampler(\+S,(\Sigma',\sigma'),w_j,r)$\; \label{alg-line-subroutine}
    }
    \For{$i\in [q]$}{
    $\rho_i\gets (\mu^{\sigma'}_{v}(i)-p_v^i)/p_{v}^0$\;
    }
    \Return $(\rho_1,\rho_2,\ldots,\rho_q)$\;
\end{algorithm}

\begin{theorem}[{\cite[Lemma B.10]{FGWWY22}}]
\label{generic-AJ}
Suppose a $q$-spin system $\+S=(G,[q],{\mathbf{b}},{\mathbf{A}})$ exhibits strong spatial mixing with decay rate $f(\ell)$, 
and there is a function $s(\ell)$ such that the neighbourhood growth of $G$ satisfies $\left|\left\{u \mid \dist_G(u,v)=\ell\right\}\right|\leq s(\ell)$ for all $v$. 
If there is some $r\in\mathbb{Z}_{\geq 1}$ such that $2eq(1+s(r))f(r)\leq 1$, 
then for any feasible boundary configuration $(\Sigma,\sigma)$, the algorithm $\ssmsampler(\+S,(\Sigma,\sigma),v,r)$ generates a sample of $v$ subject to the correct marginal distribution, 
and halts in time $O(s(r)\log \frac{1}{\varepsilon})$ with probability at least $1-\varepsilon$. 
\end{theorem}


\section{A lower bound for Weitz's algorithm}\label{sec:lb-Weitz}

In this section, we prove a lower bound for the running time of the standard implementation of Weitz's algorithm.
Consider the hard-core model on $G=(V,E)$ with parameter $\lambda$.
Suppose we want to estimate the partition function $Z$ within a \emph{constant} approximation error.
Let $V=\{v_1,\ldots,v_n\}$ and $G_i = G\setminus \{v_1,\ldots,v_{i-1}\}$.
Weitz's algorithm solves this task by estimating each $\mu_{G_i,v_i}(0)$ within an approximation error $ O(\frac{1}{n})$.  
It first constructs the SAW tree of $G_i$ rooted at $v_i$, then truncates the tree at level $\ell$ and applies dynamic programming on the truncated tree to estimate $\mu_{G_i,v_i}(0)$.
The standard implementation of Weitz's algorithm~\cite{Wei06,LLY13} ensures that for any tree with maximum degree $\Delta$, any two configurations $\sigma,\tau$ at level $\ell$, $\dTV(\mu_v^{\sigma},\mu_v^{\tau}) = O(\frac{1}{n})$. 
Standard analysis bounds the total running time from above by $T_{\mathrm{Weitz}} = \Theta(n \Delta^\ell)$.

By the same correlation decay analysis as in \Cref{lem:contraction},
when the algorithm in \Cref{small-lambda} has running time $\widetilde{O}(n^{1+1/2k})$,
we need to choose $\ell$ so that $T_{\mathrm{Weitz}}=O(n^{1+1/k})$.
This analysis only gives an upper bound on the correlation decay rate.
If the decay rate is faster, then Weitz's algorithm is faster,
and so is the algorithm in \Cref{small-lambda}.
The speedup will depend on how much faster the decay rate becomes.
Nevertheless, the next lemma shows that the analysis in \Cref{lem:contraction} is almost sharp in the worst case.
The speedup in \Cref{small-lambda} is at least $\widetilde{\Omega}\left(n^{\frac{1}{2k}-O(\frac{1}{k^2\log\Delta})}\right)$.


\begin{lemma}\label{lemma-inf-lower}
Let the real number $k > 0$ and the integer $\Delta \geq 2$ be two constants satisfying $\Delta^k \geq 4$.
Let  $\lambda = \frac{2}{(\Delta - 1)\Delta^{k}}$.
Let $T$ be an infinite $\Delta$-regular tree with root $v$.
For any $\ell \geq 2$, let $\sigma_0$ and $\sigma_1$ be all-0 and all-1 configurations at level $\ell$ of $T$ respectively.
The Gibbs distribution $\mu$ of the hard-core model on $T$ with parameter $\lambda$ satisfies 
\begin{align*}
	\dTV(\mu_v^{\sigma_0},\mu_v^{\sigma_1}) \geq \frac{1}{2}\tp{\frac{1}{\Delta^k}}^{\ell}.
\end{align*}
\end{lemma}

Let the parameters $k$, $\Delta$, and $\lambda$ be as in \Cref{lemma-inf-lower}.
%
Consider a family of hard-core instances where the graphs are indeed $\Delta$-regular trees.
%
In Weitz's algorithm, in order to ensure an $O(\frac{1}{n})$ truncation error,
\Cref{lemma-inf-lower} implies that $\ell$ must satisfy $\frac{1}{2}\tp{\frac{1}{\Delta^k}}^{\ell} = O(\frac{1}{n})$,
namely,
\begin{align*}
	\Delta^\ell = \Omega(n^{\frac{1}{k}}).
\end{align*}
This makes the overall running time $T_{\mathrm{Weitz}}  = \Omega(n^{1 + \frac{1}{k}})$.
In comparison, for these parameters, the algorithm in \Cref{small-lambda} has a running time upper bound $\widetilde{O}\left(n^{1+\frac{1}{2k}+O(\frac{1}{k^2\log\Delta})}\right)$,
which is faster by a factor of roughly $\widetilde{\Omega}(n^{1/2k})$.


\begin{proof}[Proof of Lemma~\ref{lemma-inf-lower}]
Let $w$ be an arbitrary vertex at level $0 \leq t \leq \ell$. 
Let $\pi$ denote the Gibbs distribution on the subtree rooted $T_w$ at $w$.
Recall that $\sigma_0,\sigma_1$ are pinnings on $T(\ell)$, where $T(\ell)$ is level $\ell$ of $T$.
Let $p^0_t(c) = \pi^{\sigma_0}_w(c)$ and $p^1_t(c) = \pi^{\sigma_1}_w(c)$ for $c \in \{0,1\}$, where we use $\sigma_{0}$ and $\sigma_1$ to denote all-0 and all-1 pinnings on $T_w \cap T(\ell)$.  By symmetry, $p^0_t(\cdot)$ and $p^1_t(\cdot)$ depend only on $t$ but not on $w$. In particular, $p^0_0 = \mu_v^{\sigma_0}$ and $p^1_0 = \mu_v^{\sigma_1}$ for the root $v$. For any $0 \leq t \leq \ell$, define
\begin{align*}
	R^0_t \defeq \frac{p^0_t(1)}{p^0_t(0)}, \quad R^1_t \defeq \frac{p^1_t(1)}{p^1_t(0)}.
\end{align*}

We next prove the following result holds for all $1 \leq t \leq \ell - 1$:
\begin{align}\label{eq-main}
	\abs{R^0_t  - R^1_t } \geq \frac{1}{2}\tp{\frac{1}{\Delta^k}}^{\ell  -t - 1}.
\end{align}
We need the following bound to prove~\eqref{eq-main}.
By considering the worst pinning on the neighbourhood, 
we have the following bound on both ratios $R^0_s$ and $R^1_s$
\begin{align}\label{eq-up-low-bds}
\forall 0\leq s \leq \ell - 1, \quad  R^0_s,R^1_s \leq \lambda = \frac{2}{(\Delta - 1)\Delta^{k}}.
\end{align}
We prove~\eqref{eq-main} by induction on $t$ from $\ell - 1$ to 1.
The base case is $t = \ell - 1$. 
Note that $\Delta^k \geq 4$.
A straightforward calculation shows that
\begin{align*}
	\abs{R^0_{\ell - 1} - R^1_{\ell - 1}} = \abs{ 1 - \lambda} = 1 - \frac{2}{(\Delta - 1)\Delta^{k}} \geq \frac{1}{2}.
\end{align*}
For the induction step, fix $1 \leq t \leq \ell - 2$. 
The recursion function in $(\Delta-1)$-ary tree is
\begin{align*}
	f(x) = \lambda \tp{\frac{1}{1+x}}^{\Delta - 1}.
\end{align*}
Note that $R^0_t = f(R^0_{t+1})$ and $R^1_t = f(R^1_{t+1})$. By the mean value theorem, there exists  $\theta$ such that $\min(R^0_{t+1}, R^1_{t+1}) < \theta < \max(R^0_{t+1}, R^1_{t+1}) $ and 
\begin{align*}
	\abs{R^0_t  - R^1_t } = |f'(\theta)| \cdot  \abs{R^0_{t+1} - R^1_{t+1}}. 
\end{align*}
By~\eqref{eq-up-low-bds} and the fact $\Delta^k \geq 4$, we have
\begin{align*}
	\abs{f'(\theta)} &= \lambda(\Delta - 1) \tp{\frac{1}{1+\theta}}^{\Delta} \geq \lambda(\Delta - 1)\tp{\frac{1}{1+\lambda}}^{\Delta} \geq \lambda(\Delta - 1)\exp(-\lambda\Delta)\\
	&= \frac{2}{\Delta^k} \exp \tp{ - \frac{2\Delta}{(\Delta - 1)\Delta^k} } \geq \frac{2}{\Delta^k} \exp \tp{ - \frac{4}{\Delta^k} } \geq \frac{1}{\Delta^k}.
\end{align*}
By the induction hypothesis that $\abs{R^0_{t+1} - R^1_{t+1}} \geq \frac{1}{2} (\frac{1}{\Delta^k})^{\ell - t - 2} $, we can prove~\eqref{eq-main} for $t$. This finishes the induction step for $1 \leq t \leq \ell - 3$.

Finally, we use~\eqref{eq-main} to bound $|R^0_0 - R^1_0|$. 
The proof is similar to the proof in the induction step.
The only difference is that the recursion for root $v$ becomes $g(x) = \lambda \tp{\frac{1}{1+x}}^{\Delta}$. By a similar calculation, there exists $\min(R^0_{1}, R^1_{1}) < \theta < \max(R^0_{1}, R^1_{1}) $ such that 
\begin{align*}
|R^0_0 - R^1_0| &= |g'(\theta)| \cdot  \abs{R^0_{1} - R^1_{1}}  \geq  \Delta \lambda \tp{\frac{1}{1+\theta}}^{\Delta+1} \cdot 	\frac{1}{2}\tp{\frac{1}{\Delta^k}}^{\ell-2}\\
&= \frac{\Delta}{(\Delta - 1)(1+\theta)} \cdot \lambda (\Delta-1) \tp{\frac{1}{1+\theta}}^{\Delta}  \cdot \frac{1}{2}\tp{\frac{1}{\Delta^k}}^{\ell-2}\\
&\geq  \frac{\Delta}{(\Delta - 1)(1+\lambda)} \cdot \frac{1}{2}\tp{\frac{1}{\Delta^k}}^{\ell-1}
 \geq \frac{1}{3}\tp{\frac{1}{\Delta^k}}^{\ell-1}.
\end{align*}

By the definitions of $R^0_0$ and $R^1_0$ and the fact $\Delta^k \geq 4$, we have
\begin{align*}
	\dTV\tp{ \mu^{\sigma_0}_v, \mu^{\sigma_1}_v } &= \abs{\mu^{\sigma_0}_v(1) - \mu^{\sigma_1}_v(1)} = \mu^{\sigma_0}_v(0)\mu^{\sigma_1}_v(0) |R^0_0 - R^1_0|\\
	&\geq \tp{\frac{1}{1+\lambda}}^2 \cdot \frac{1}{3}\tp{\frac{1}{\Delta^k}}^{\ell-1} 
  \geq \frac{4}{27}\tp{\frac{1}{\Delta^k}}^{\ell-1}
	\geq \frac{1}{2}\tp{\frac{1}{\Delta^k}}^{\ell}. \qedhere
\end{align*}
\end{proof}

\section{Grid graph with quadratic-sized boundary}\label{sec:quadratic-boundary}

Although the distance-$n$ boundary of the $\mathbb{Z}^2$ lattice contains only $4n$ vertices, 
there are subgraphs that blow this number up to $\Omega(n^2)$. 
Below is one such graph constructed recursively. 
For large enough even $n$, the graph $G(n)$ is given by \Cref{quad-boundary-grid}. 
\begin{figure}[!htbp]
\centering
	\includegraphics[width=\textwidth]{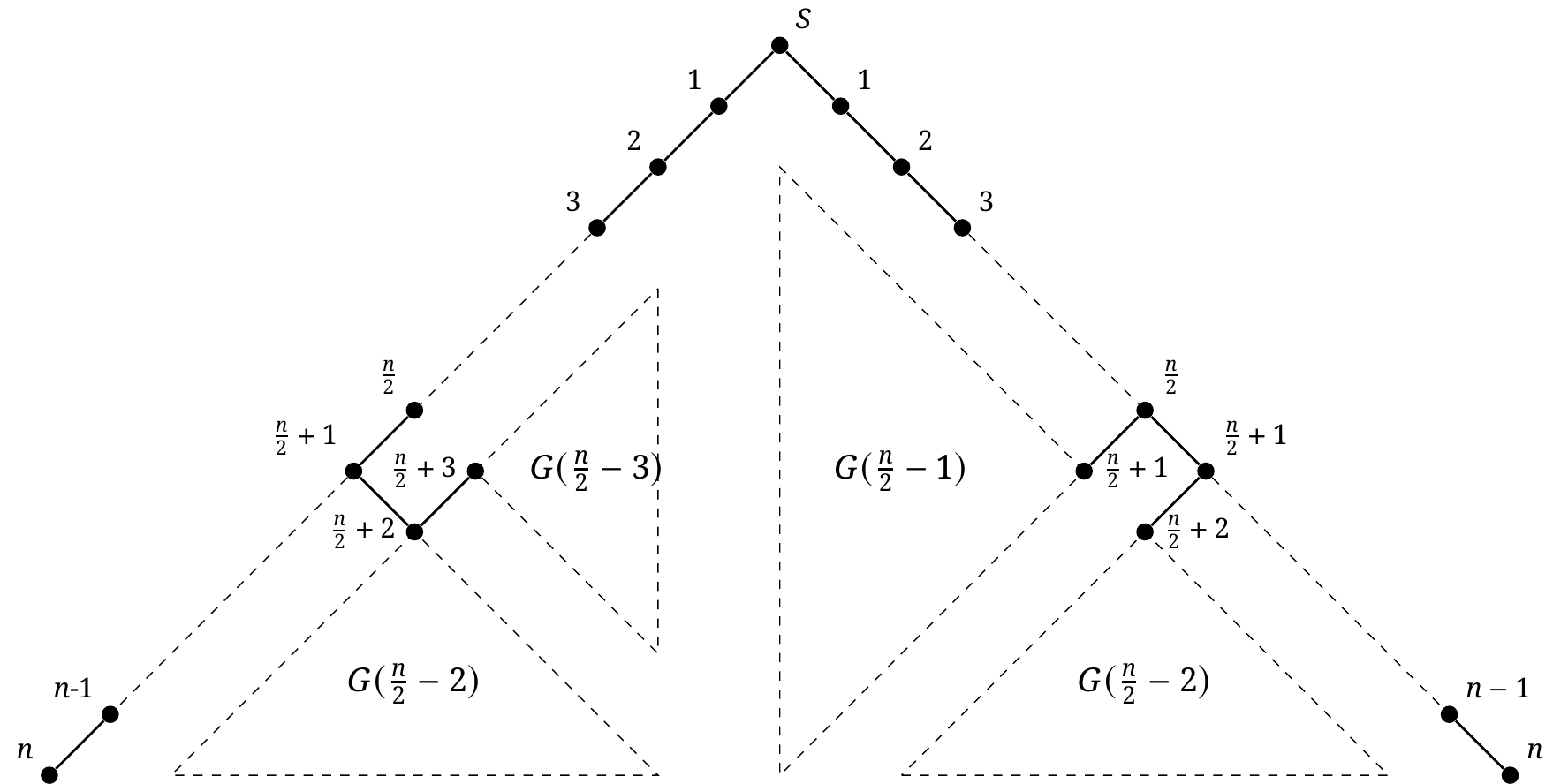}
\caption{The graph $G(n)$ when $n$ is even. The number next to a vertex indicates its distance from the starting vertex $S$. }
\label{quad-boundary-grid}
\end{figure}

The odd-$n$ case can be constructed similarly. 
Let $f(n)$ be the number of distance-$n$ vertices from the vertex $S$ in the graph $G(n)$. 
Then $f(n)=4f(\frac{n}{2}-\Theta(1))+\Theta(1)$.
By the Master Theorem, $f(n)=\Theta(n^2)$. 
Also note that such a construction can be made on an induced subgraph of $\mathbb{Z}^2$, by splitting each edge here into two edges joined by a vertex. 

\end{document}